\documentclass[sigconf]{acmart}
\AtBeginDocument{%
  \providecommand\BibTeX{{%
    \normalfont B\kern-0.5em{\scshape i\kern-0.25em b}\kern-0.8em\TeX}}}

\setcopyright{none}
\renewcommand\footnotetextcopyrightpermission[1]{}
\usepackage[list=true]{subcaption}
\usepackage{tabularx}
\usepackage[T1]{fontenc}
\usepackage{amsmath}
\usepackage{cleveref}
\usepackage{amsthm}
\usepackage{amsfonts}
\usepackage{amssymb}
\usepackage{dsfont}
\usepackage{multirow}
\usepackage{diagbox}
\usepackage{graphicx}
\usepackage{ifthen}
\graphicspath{{/.images/}}

\def\LLR{\mathsf{LLR}}

\def\etr{\mathsf{etr}}

\newcounter{casenum}
\newenvironment{caseof}{\setcounter{casenum}{1}}{\vskip.5\baselineskip}
\newcommand{\case}[2]{\vskip.5\baselineskip\par\noindent {\bfseries Case \arabic{casenum}:} #1\\#2\addtocounter{casenum}{1}}

\newcommand{\ra}{\rightarrow}

\newcommand{\imp}{\Rightarrow}

\newcommand{\indicator}[1]{ \boldsymbol{1}_{\{#1\}}}

\DeclareMathOperator*{\argmax}{arg\,max}

\newcommand{\R}{\mathbb{R}}

\newcommand{\ignore}[1]{}

\newcommand{\X}{\mathcal{X}}
\newcommand{\Exp}[1]{\mathbb{E}\left[#1\right]}

\newboolean{showcomments} \setboolean{showcomments}{false}
\newcommand{\jk}[1]{ \ifthenelse{\boolean{showcomments}}
  {\textcolor{red}{(JK says: #1)}} {} } \newcommand{\vd}[1]{
  \ifthenelse{\boolean{showcomments}} {\textcolor{red}{(VD says: #1)}}
  {} }

\begin{document}

\title{Pareto-optimal energy sharing between battery-equipped
  renewable generators}

\author{Vivek Deulkar}
\email{{vivekdeulkar, jayakishnan.nair}@iitb.ac.in}
\orcid{1234-5678-9012}
\author{Jayakrishnan Nair}
\affiliation{%
 \institution{IIT Bombay}
}
\renewcommand{\shortauthors}{Vivek Deulkar \& Jayakrishnan Nair}

\begin{abstract}
  The inherent intermittency of renewable sources like wind and solar
  has resulted in a bundling of renewable generators with storage
  resources (batteries) for increased reliability. In this paper, we
  consider the problem of energy sharing between two such bundles,
  each associated with their own demand profiles. The demand profiles
  might, for example, correspond to commitments made by the bundle to
  the grid. With each bundle seeking to minimize its loss of load
  rate, we explore the possibility that one bundle can supply energy
  to the other from its battery at times of deficit, in return for a
  reciprocal supply from the other when it faces a deficit itself. We
  show that there always exist \emph{mutually beneficial} energy
  sharing arrangements between the two bundles. Moreover, we show that
  Pareto-optimal arrangements involve at least one bundle transferring
  energy to the other at the maximum feasible rate at times of
  deficit. We illustrate the potential gains from such dynamic energy
  sharing via an extensive case study.
\end{abstract}



\maketitle

\section{Introduction}
\label{sec:intro}

Environmental concerns are driving a worldwide push towards the
adoption of renewable sources for electricity generation. However, the
primary challenge in increasing the penetration of renewable sources
in the electricity grid is the intermittency and unpredictability of
their generation. As a result, renewable sources are increasingly
being bundled with storage devices (batteries) to smoothe out the
temporal intermittency in generation, enabling the generator to supply
power to the grid according to a contracted supply profile with high
reliability. In this paper, we explore \emph{dynamic energy sharing}
between two such bundles (also referred to as agents), each consisting
of a renewable generator and a battery, with the goal of enhancing the
reliability of both bundles.

The idea behind dynamic energy sharing is that one can exploit
statistical diversity between the net generation (supply minus demand)
processes of the bundles, so that one agent can supply energy to the
other (from its battery) in times of deficit. For example, a solar
generator and a wind generator could opportunistically supply energy
to one another, in order to meet their respective commitments to the
grid. However, such a sharing arrangement would only work if it
benefits both parties involved, i.e., the sharing arrangement must be
\emph{mutually beneficial}. Interestingly, naive complete pooling,
wherein the two batteries are treated as one common resource, may not
have this property. This motivates us to consider \emph{partial
  sharing mechanisms}, wherein each agent supplies energy to the other
to meet its deficit, but only up to a pre-specified drain rate (i.e.,
an upper bound on the rate of energy supply). One of our main
contribution is to show that there always exist mutually beneficial
partial sharing mechanisms of this kind.

Having established the existence of mutually beneficial sharing
configurations, the next natural step is to capture the \emph{Pareto
  frontier} of efficient sharing configurations. Remarkably, we are
able to provide a precise characterization of this Pareto frontier;
all Pareto-optimal configurations involve at least one agent allowing
the other to draw energy at the maximum possible rate in times of
deficit. Given this characterization of the Pareto frontier, one can
capture the sharing arrangements that would emerge between the agents
by appealing to the theory of bargaining \cite{Myerson1991}.

Structurally, our work is related to the vast literature on resource
pooling in service systems and networks. This body of work explores
the sharing of a resource, such as service capacity, bandwidth, cache
memory, etc., to exploit statistical economies of
scale~\cite{singh2011cooperative,
  sarkar2008coalitional,grokop2008spectrum,huang2006auction,
  dahlin1994cooperative, cardellini1999dynamic, archer2006load}.
However, the primary style of pooling considered in this literature is
\emph{complete pooling}, wherein the entities `merge' by pooling their
resources completely, with the payoff of the coalition being split
among its members using ideas from cooperative game theory. In
contrast, our interest here is in the setting of
\emph{non-transferable utility}, wherein no side payments occur
between the agents, implying that an agent would agree to a sharing
arrangement only if doing so increases its own (private,
non-transferable) utility. The only other work we are aware of that
takes this approach is~\cite{Nandigam19}, which focuses on partial
server sharing between two Erlang-B loss systems. However, the system
modeling, sharing mechanism, and its analysis differ considerably
between~\cite{Nandigam19} and the present paper.

At its core, our model involves inventory pooling between two agents,
each having its own supply and demand process. These agents might
correspond to renewable generators connected to the grid with their
own contracted supply commitments, prosumers within a smart microgrid
(as in \cite{zhu2013sharing}), or even different remote microgrids
that can help each another increase their reliability\footnote{\tiny
  \url{https://www.cleantech.com/power-to-the-people-remote-microgrids-across-southeast-asia/}}. Our
main departure from the prior literature on these topics (described
below) is the cooperative game theoretic framework of
non-transferable utility. This is applicable in situations where
agents cannot balance a certain loss of reliability with a monetary
reward. Instead, we appeal to the theory of bargaining to balance the
benefit of both agents from the sharing arrangement. For example, in
the case of renewable generators participating in electricity markets,
maintaining high reliability might be a pre-requisite for market
participation. Similarly, remote microgrids might prefer not to set up
a system of monetary payments, but to simply maintain equity in the
benefits from energy sharing.

The remainder of this paper is organised as follows. After a quick
review of the related literature below, we describe our system model
and our energy sharing mechanism in Section~\ref{sec:model}. We
establish key monotonicity properties of our sharing model in
Section~\ref{sec:monotonocity}. Pareto-optimal sharing configurations
are addressed in Section~\ref{sec:pareto}. We present a case study in
Section~\ref{sec:case_study}, and conclude in
Section~\ref{sec:discuss}.

\subsection*{Related literature}

There are two broadly two strands of work on dynamic energy sharing
between renewable generators (or prosumers). The first treats energy
trading as an economic transaction, using prices to design market
mechanisms for efficient energy transfer; papers that take this
approach include
\cite{liu2017energy,dimeas2005operation,Cintuglu15,eddy2014multi,
  luo2014autonomous,carrion2009bilevel,lee2014direct}.

The second strand of work treats energy sharing from the standpoint of
a central optimizer, who is interested in maximizing social
welfare. For example, \cite{lakshminarayana2014cooperation} analyzes
the tradeoff between sharing and the use of storage within smart
microgrid. Another work in this space is \cite{zhu2013sharing}, which
focuses on scheduling algorithms for energy transfer within a
microgrid for loss minimization.

In contrast to the first stream of work, the present paper considers
the non-transferable utility setting, i.e., one where there are no
side-payments between the agents. In contrast to the second strand, we
still consider the agents as strategic, in that they participate in
energy sharing mechanisms with the sole objective to enhancing their
own reliability.

\section{Model and Preliminaries}
\label{sec:model}

We begin by describing the `standalone' setting, i.e., in the absence
of a sharing arrangement between the agents. We use the following
notation throughout the paper: For $x \in \R,$
\begin{align*}
  [x]_{+}&= (x)_{+} =\max(0,x),\\
  [x]_{-}&= (x)_{-} =-\min(0,x).
\end{align*}
Also, when referring to any Agent~$i,$ we refer to the \emph{other}
agent as Agent~$-i.$

\subsection{Standalone setting}

We consider two agents, Agent~1 and Agent~2, each associated with a
stochastic supply and demand process, and equipped with a battery.
The battery of Agent~$i$ has capacity $B_i.$ The energy content in the
battery is modulated by each agent's net generation (i.e., supply
minus demand) process. Specifically, we denote the net (power)
generation process of Agent~$i$ by $r_i(t) = g_i(t) - d_i(t),$ where
$g_i(t)$ and $d_i(t)$ denote, respectively, the supply and demand
process of Agent~$i.$ For example, $g_i(t)$ might represent the power
generation from a renewable generator, and $d_i(t)$ its contracted
commitment to the grid. Thus, the energy content in the battery of
Agent~$i,$ denoted by $b_i(t),$ evolves as follows:
\begin{equation}
\frac{d}{dt}b_i(t) = 
    \begin{cases}
      
      0           & \text{if }b_i(t) = 0 \text{ and } r_i(t)<0,\\ 
      0           & \text{if }b_i(t) = B_i \text{ and } r_i(t)>0,\\
      r_i(t)    & \text{ otherwise. }      
    \end{cases} \label{eq:bdot}
\end{equation}
The above dynamics capture the boundary conditions that a fully
charged battery cannot be charged further, and an empty battery cannot
be discharged further.\footnote{In practice, for certain battery
  chemistries, it is not advisable to drain/charge the battery
  completely. So `empty' and `full' in our model could in practice
  correspond to the floor and ceiling of charge level the agent
  chooses to operate the battery with.} Except for these boundary
cases, the battery is charged/discharged at the net generation rate
$r_i(t).$ Further, we assume that the net generation rates $r_1(t)$
and $r_2(t)$ are dependent on the state of a common background Markov
process, which captures the supply and demand uncertainty of each
agent.

Formally, let $\{X(t)\}$ denote the background Markov process. We
assume that that $\{X(t)\}$ is an irreducible continuous-time Markov
chain (CTMC) over a finite state space $S.$ The net generation rate of
Agent~$i$ is a function of the state of this background process, i.e.,
$r_i(t) = r_i(X(t)).$

To ensure regeneration of the buffer occupancy processes, we assume
there exist in $S,$ states $s_1,\ s_2,\ s_3,$ and $s_4$ satisfying:
\begin{itemize}
\item $r_1(s_1),r_2(s_1) > 0,$
\item $r_1(s_2),r_2(s_2) < 0,$
\item $r_1(s_3) > 0,\ r_2(s_3) < 0,$ and
  $r_2(s_3) \leq r_2(s)\ \forall s \in S,$
\item $r_1(s_4) <0,\ r_2(s_4) > 0,$ and
  $r_1(s_4) \leq r_1(s)\ \forall s \in S.$
\end{itemize}
State~$s_1$ (respectively, $s_2$) results in positive (respectively,
negative) net generation for both agents. State~$s_3$ results in a
positive net generation for Agent~1 and a negative net generation for
Agent~2, such that this is also the most negative net generation for
Agent~2. State~$s_4$ is similar, except the roles of the agents are
reversed.

We note here that this model is quite general; the state of the
background process could capture all factors that influence the supply
and demand of each agent, including past and present weather
conditions, as well as time of day. Moreover, the model allows for the
net generation processes of both agents to be correlated in an
arbitrary fashion. Finally, ramp constraints on battery
charging/discharging can also be incorporated into the model; these
would simply limit the values that the (battery modulating) net
generation process can take.

The performance of each agent is measured via its \emph{loss of load
  rate} ($\LLR$), which is the long run average rate of lost load
(unmet demand).
Based on the battery
dynamics~\eqref{eq:bdot}, note that Agent~$i$ is unable to cater to
its demand when $b_i(t) = 0$ and $r_i(t) < 0,$ i.e., the battery is
empty and the instantaneous generation is insufficient to meet the
instantaneous demand.  Thus, the standalone loss of load rate of
Agent~$i$ is defined as
\begin{equation}
  \label{eq:llr_standalone}
  \LLR_i^{sa} = \lim_{t \ra \infty} \frac{1}{t} \int_0^t \indicator{b_i(t) = 0} [r_i(t)]_{-} dt \quad \text{(almost surely)}.
\end{equation}
Here, $\indicator{A}$ equals 1 if $A$ is true and zero otherwise.
Each agent seeks to minimize its loss of load rate. Indeed, it would
be natural for the grid to penalize a renewable generator in terms of
the $\LLR$ relative to its contracted supply curve
\cite{EWEA_energy_imbalance_penalty, brunetto2011wind}. Note that the
existence of the almost sure limit in \eqref{eq:llr_standalone}
follows from the positive recurrence of the background Markov process.

For any agent, the evolution of its battery level can be modeled as a
\emph{Markov modulated fluid queue}, a well studied object in the
queueing literature \cite{anick1982stochastic, Mitra88}. In
particular, it is easy to see that corresponding to Agent~$i,$
$(X(t),b_i(t))$ is a Markov process over state space
$S \times [0,B_i],$ whose invariant distribution can be computed by
solving a system of ordinary differential equations (see
\cite{anick1982stochastic, Mitra88}). This invariant distribution can
in turn be used to compute $\LLR_i$~\cite{deulkar2019sizing}. However,
in the present paper, we are interested not in the `standalone'
behavior of each individual agent, but in \emph{dynamic energy
  sharing} arrangements between the two. (This naturally couples the
evolution of the two batteries, necessitating a joint analysis of both
battery occupancies.) Our proposed mechanism for dynamic energy
sharing is described next.

\subsection{Sharing Mechanisms}
\label{subsec:sharing_strategy}

For either agent, loss of load is undesirable (and might even result
in penalty from the grid \cite{EWEA_energy_imbalance_penalty,
  brunetto2011wind}). The motivation for dynamic energy sharing is
that when Agent~$i$ faces loss of load, Agent~$-i$ can supply energy
to Agent~$i$ from its battery to (either completely or partially)
satisfy the unmet demand.
However, if no constraints are placed on such energy transfer, i.e.,
if each agent can draw energy from the other's battery without
restriction, the resulting configuration might not be beneficial to
both parties. Indeed, if Agent~$i$ is considerably more likely to run
a deficit than Agent~$-i,$ it is natural to expect that unconstrained
energy sharing would actually be detrimental to
Agent~$-i.$\footnote{That unconstrained energy sharing is not
  guaranteed to be mutually beneficial will be demonstrated later.}
This motivates us to explore \emph{constrained} or \emph{partial}
energy sharing arrangements.

\begin{table*}[htp]
 \begin{center}
 	\begin{tabular}{ ||c| c| c|c|| }
 	\hline \hline
 	\diagbox{$b_1$} {$b_2$} &\text{Empty} & \text{Int.} & \text{Full}\\
 	\hline \hline
 	\multirow{2}{*}{Empty}	
 	&$\frac{d}{dt}b_1=[r_1-\min(c_1,(r_2)_-)]_+$  &$\frac{d}{dt}b_1=\max(0,r_1)$  &$\frac{d}{dt}b_1=\mathds{1}_{\{r_2 \geq \min(c_2,(r_1)_-)\}} [r_1 + \min(c,(r_2)_+)]_+$\\
 	&$\frac{d}{dt}b_2=[r_2-\min(c_2,(r_1)_-)]_+$  &$\frac{d}{dt}b_2=r_2-\min(c_2,(r_1)_-)$ &$\frac{d}{dt}b_2= \mathds{1}_{\{r_2 < \min(c_2,(r_1)_-)\}} [r_2-\min(c_2,(r_1)_-)]$\\
 	\hline
 	\multirow{2}{*}{Int.}
 	&$\frac{d}{dt}b_1=r_1-\min(c_1,(r_2)_-)$ &$\frac{d}{dt}b_1=r_1$ &$\frac{d}{dt}b_1=r_1+\min(c,(r_2)_+)$\\
 	&$\frac{d}{dt}b_2=\max(0,r_2)$ &$\frac{d}{dt}b_2=r_2$ &$\frac{d}{dt}b_2=\min(0,r_2)$\\
 	\hline
 	\multirow{2}{*}{Full}
 	&$\frac{d}{dt}b_1=\mathds{1}_{\{r_1<\min(c_1,(r_2)_-)\}} [r_1-\min(c_1,(r_2)_-)]$ &$\frac{d}{dt}b_1=\min(0,r_1)$ &$\frac{d}{dt}b_1=\min(0,r_1)$\\
 	&$\frac{d}{dt}b_2=\mathds{1}_{\{r_1 \geq \min(c_1,(r_2)_-)\}} [r_2 + \min(c,(r_1)_+)]_+$ &$\frac{d}{dt}b_2=r_2+\min(c,(r_1)_+)$ &$\frac{d}{dt}b_2=\min(0,r_2)$\\
 	\hline \hline
 	
 	\end{tabular}
\end{center}
\caption{Battery dynamics of two agents; the battery level
  $\boldsymbol{b_1}$ of Agent~1 is along the rows while the battery
  level $\boldsymbol{b_2}$ of Agent~2 is along the columns. Here Int. stands for intermediate.}
\label{table:bat_dynamics}
\end{table*}

Our energy sharing arrangement is characterized by the tuple
$(c_1,c_2).$ Informally, $c_i$ is the maximum energy drain rate
allowed by Agent~$i$ when Agent~$-i$ faces loss of load. Additionally,
we assume that there is a capacity constraint $c$ on the rate of
energy transfer from either agent to the other. This captures any
physical transmission constraints that limit the rate of energy
transfer between the agents.\footnote{For simplicity, we take the
  capacity constraint to be symmetric; our results extend naturally to
  the case where there are different upper bounds on the energy
  transfer rate in the two directions. Moreover, in the absence of
  such constraints, $c$ may be set to be a suitably large value
  (say $\max_{x \in S} \left\{ \max(|r_1(x)|, |r_2(x)|) \right\}$)
  such that this constraint is never binding.} Thus,
$c_1,\ c_2 \leq c.$

Formally, if Agent~$-i$ runs a deficit at time~$t,$ i.e.,
$b_{-i}(t) = 0$ and $r_{-i}(t) < 0,$ then:
\begin{enumerate}
\item If $b_i(t) \in (0,B_i),$ then Agent~$i$ transfers energy to
  Agent~$-i$ at rate $$\etr_i(t) = \min(c_i,-r_{-i}(t)).$$ Here,
  $\etr$ stands for \emph{energy transfer rate}. In other words,
  Agent~$i$ helps cut the deficit rate of Agent~$-i$ subject to a
  maximum transfer rate of $c_i.$
\item If $b_i(t) = 0$ and $r_i(t) > 0,$ then Agent~$i$ transfers
  energy to Agent~$-i$ at rate
  $$\etr_i(t) = \min(r_i(t),c_i,-r_{-i}(t)).$$ In this case, since
  Agent~$i$ cannot draw energy from its battery, its ability to
  transfer energy to Agent~$-i$ is further limited by its own net
  generation rate $r_i(t)$. If $b_i(t) = 0$ and $r_i(t) \leq 0,$ then
  Agent~$i$ itself faces loss of load, and therefore cannot transfer
  energy to Agent~$-i,$ i.e., $\etr_i(t) = 0.$
\item if $b_i(t) = B_i,$ then Agent~$i$ transfers energy to Agent~$-i$
  at rate
  \begin{itemize}
  \item $\etr_i(t) = \min(c_i,-r_{-i}(t))$ if $r_i \leq \min(c_i,-r_{-i}(t)),$ and
  \item $\etr_i(t) = \min(c,r_i)$ if $r_i > \min(c_i,-r_{-i}(t)).$
  \end{itemize}
  In the former case, the nominal transfer rate of
  $\min(c_i,-r_{-i}(t))$ exceeds the net generation rate of Agent~$i,$
  and so Agent~$i$ begins to discharge its battery to help Agent~$-i.$
  In the latter case, the net generation rate of Agent~$i$ exceeds the
  nominal transfer rate, and so Agent~$i$ does not discharge its
  battery, but simply transfers its surplus generation to Agent~$-i$
  subject to the capacity constraint $c.$
\end{enumerate}
Additionally, when Agent~$i$ has a fully charged battery with a
positive net generation rate, we assume that it transfers its overflow
rate to Agent~$-i$ (whether or not Agent~$-i$ faces loss of load),
again subject to the capacity constraint $c.$ Of course, if both
batteries are fully charged, then the overflow is lost. The above
sharing mechanism is summarized in Table~\ref{table:bat_dynamics}.

It is important to note that under the proposed mechanism, unless
Agent~$i$'s battery is full, it only transfers energy to Agent~$-i$
when the latter faces loss of load. Moreover, the energy transfer in
this case only covers (part of) the deficit rate, it does not actually
charge the battery of Agent~$-i.$

Next, we define the loss of load rate for each agent under the
proposed sharing mechanism.

\subsection{$\LLR$ characterization}

Under the sharing configuration $(c_1,c_2),$ the loss of load rate of
each agent is characterized as follows. When Agent~$i$ faces loss of
load, i.e., $b_{i}(t) = 0$ and $r_{i}(t) < 0,$ then the instantaneous
rate at which it loses load equals $[-r_{i}(t) - \etr_{-i}(t)]_+.$
Thus, its loss of load rate is given by
$$\LLR_{i}(c_1,c_2) = \lim_{t \ra \infty} \frac{1}{t} \int_0^t
\indicator{b_{i}(t) = 0,\ r_{i}(t) < 0}[-r_{i}(t) - \etr_{-i}(t)]_+ \
dt.$$ As before, the above limit is in an almost sure sense, and its
existence is guaranteed by the positive recurrence of the background
process.

It is important to note that the sharing configuration $(0,0)$ is not
equivalent to the standalone setting, since even under the $(0,0)$
configuration, the agents supply (overflow) energy to one another when
their batteries are full. In fact, it can be shown that
$\LLR_i(0,0) < \LLR_i^{sa}$ for $i = 1,2.$

Given the complicated form of the sharing mechanism, one would not
expect a closed form characterization for $\LLR_{i}(c_1,c_2);$ indeed
a closed form for the loss of load rate does not exist even for the
`standalone' setting.
However, we are still remarkably able to analytically establish the
following results.
\begin{itemize}
\item There exist mutually beneficial sharing
  configurations. Specifically, there exists a configuration
  $(c_1,c_2) \in [0,c]^2$ such that
  $\LLR_i(c_1,c_2) < \LLR_i(0,0) < \LLR_i^{sa}$ for $i = 1,2.$
\item The Pareto frontier of efficient, mutually beneficial, sharing
  configurations is non-empty. All configurations on the Pareto
  frontier involve at least one agent sharing energy with the other at
  the maximum possible rate, i.e., $c_i = c$ for at least one
  $i \in \{1,2\}.$
\end{itemize}
The above results are proved using monotonicity properties of the loss
of load rates with respect to $c_1$ and $c_2.$ These monotonicity
properties, which are illuminating in their own right, are the focus
of the following section.

\section{Monotonicity Properties}
\label{sec:monotonocity}

Without loss of generality, we assume that
$$c_i \leq c_{i,\max}
:= \min(c,\ \max_{s \in S} \{[r_{-i}(s)]_{-}\}$$ for $i \in \{1,2\}).$
Given the sharing mechanism described in
Section~\ref{subsec:sharing_strategy}, it is easy to see that
increasing $c_i$ beyond $\max_{s \in S} \{[r_{-i}(s)]_{-}\}$
does not influence the realised energy sharing, since agents only help
fulfill the other's deficit rates (except under overflow, which is not
constrained by $(c_1,c_2)$).

The main result of this section is the following.
\begin{theorem}
  \label{thm:monotonocity}
  Under the proposed sharing mechanism,
  \begin{itemize}
  \item $\LLR_i(c_1,c_2)$ is strictly increasing in $c_i$ over $[0,c_{i,\max}],$
  \item $\LLR_i(c_1,c_2)$ is strictly decreasing in $c_{-i}$ over
    $[0,c_{-i,\max}],$
  \item $\LLR_1(c_1,c_2) + \LLR_2(c_1,c_2)$ is strictly decreasing in
    $c_i$ over $[0,c_{i,\max}].$
\end{itemize}
\end{theorem}

The first two statements of Theorem~\ref{thm:monotonocity} are
intuitive; if the peak rate $c_i$
of energy transfer from Agent~$i$ to Agent~$-i$ is increased, the loss
of load rate of Agent~$i$ increases, while that of Agent~$-i$
decreases. The third statement shows that an increase in $c_i$
decreases the overall loss of load rate across the two agents. In
other words, an increase in $c_i$ is detrimental to Agent~$i,$
beneficial to Agent~$-i,$ and beneficial to the overall system $\LLR.$
Immediate takeaways from Theorem~\ref{thm:monotonocity} are the
following. First, the `socially optimal' configuration is
$(c_{1,\max},c_{2,\max}).$ Second, interpreting $c_i$ to be
Agent~$i$'s action, the only Nash equilibrium of the resulting
two-player game is (0,0). This means that a non-cooperative setting
does not yield efficient sharing configurations. Instead, mutually
beneficial sharing configurations can only be sustained via
\emph{binding agreements} between the agents, in the spirit of
bargaining theory.

The remainder of this section is devoted to the proof of
Theorem~\ref{thm:monotonocity}. We are able to prove monotonicity
properties of the loss of load rates without any explicit expressions
of the same using sample path techniques. Specifically, we consider
two identical instances of our model, one operating under sharing
configuration $(c_1,c_2),$ where $c_1 < c_{1,\max}$ and the other
operating under sharing configuration $(c_1+\epsilon, c_2),$ where
$\epsilon>0$ such that $c_1 + \epsilon \leq c_{1,\max}.$ We refer to
the former system as the `original system' and the latter one as the
`$\sim$ system'; we denote parameters pertaining to this latter system
with a $\sim$ accent. To prove Theorem~\ref{thm:monotonocity}, it
suffices to show that
\begin{align*}
  \LLR_1 > \widetilde{\LLR_1}, \LLR_2 < \widetilde{\LLR_1}, \text{ and }  
  \LLR_1 + \LLR_2 <  \widetilde{\LLR_1} + \widetilde{\LLR_2}.
\end{align*}
We prove these inequalities by coupling the sample paths of the
background process across these systems. In other words, both systems
see exactly the same net generation at all times;\footnote{That is,
  the battery levels $b_1 (t),\ \tilde{b}_1(t)$ get modulated by the
  same net generation process $r_1(t)$ while battery levels
  $b_2(t),\ \tilde{b}_2(t)$ get modulated by same net generation
  process $r_2(t).$} all that distinguishes these systems is the upper
bound on the energy transfer rate from Agent~1 to Agent~2 ($c_1$ in
the original system versus $c_1 + \epsilon$ in the $\sim$ system). Our
first result compares the battery levels of the two agents across the
two systems.

\begin{lemma}
  {\label{lem:bat_induction}} On any sample path, under the coupling
  between the original and $\sim$ system described above, if 
  $b_1(0) = \tilde{b}_1(0),\ b_2(0) = \tilde{b}_2(0),$ then for all
  $t \geq 0,$
  \begin{equation}\label{eq:bat_induction}
    \tilde{b}_1(t) \leq b_1(t),\quad \tilde{b}_2(t)\leq b_2(t).
  \end{equation} 
\end{lemma}
Lemma~\ref{lem:bat_induction} states that in the $\sim$ system, which
permits a higher rate of energy transfer from Agent~1 to Agent~2
relative to the original system, the battery occupancies of
\emph{both} agents get reduced. The intuition behind this result is
the following. Since Agent~1 transfers more energy to Agent~2 in the
$\sim$ system, its battery occupancy gets reduced relative to the
original system. Importantly however, this increased energy transfer
from Agent~1 to Agent~2 is utilized purely to cut the loss of load
suffered by Agent~2, and \emph{not} to charge its battery. Over time,
this causes a reduction in the battery occupancy of Agent~2 as well,
due to (i) reduced energy transfer from Agent~1 in the form of battery
overflow (which \emph{can} be used to charge Agent~2's battery), and
(ii) increased energy transfer to Agent~1 when it faces loss of
load. This intuition is formalized in our proof of
Lemma~\ref{lem:bat_induction}, which can be found in the appendix.

An immediate consequence of Lemma~\ref{lem:bat_induction} is the
following lemma. Let $\mathcal{O}_i(t)$ (respectively,
$\tilde{\mathcal{O}}_i(t)$) denote the cumulative energy \emph{lost}
due to battery overflow in the interval $[0,t]$ from the battery of
Agent~$i$ in the original system (respectively, in the $\sim$
system). Also, let $\ell_i(t)$ (respectively, $\tilde{\ell}_i(t)$)
denote the cumulative lost load by Agent~$i$ in the original system
(respectively, in the $\sim$ system). Specifically, note that for
$i \in \{1,2\},$
\begin{align*}
  \LLR_i = \lim_{t \ra \infty}\frac{\ell_i(t)}{t}, \quad
  \widetilde{\LLR}_i = \lim_{t \ra \infty}\frac{\tilde{\ell}_i(t)}{t}. 
\end{align*}

\begin{lemma}{\label{lem:LOL_overflow}}
  On any sample path, under the coupling between the original and
  $\sim$ system described above, if
  $b_1(0) = \tilde{b}_1(0),\ b_2(0) = \tilde{b}_2(0),$ then for all
  $t \geq 0,$
  \begin{align}
    \label{eq:lol_comparison}
    \tilde{\ell}_1(t) &\geq \ell_1(t), \\
    \label{eq:overflow_comparison}
    \tilde{\mathcal{O}}_i(t)  &\leq  \mathcal{O}_i(t) \quad \forall\ i\in \{1,2\}.
  \end{align} 
\end{lemma}
Lemma~\ref{lem:LOL_overflow} states that the $\sim$ system `wastes'
less energy due to overflow from either battery, as compared to the
original system. Further, it shows that in the $\sim$ system, Agent~1
incurs greater loss of load compared to the original system.
\begin{proof}[Proof of Lemma~\ref{lem:LOL_overflow}]
  Since $\tilde{b}_1(t) \leq b_1(t)$ (Lemma~\ref{lem:bat_induction}),
  if Agent~1 faces loss of load in the original system at any
  time~$t$, then it also faces loss of load in the $\sim$ system at
  that time. Moreover, since (i) $\tilde{b}_2(t) \leq b_2(t)$ (also by
  Lemma~\ref{lem:bat_induction}), and (ii) the bound $c_2$ on the rate
  of energy transfer from Agent~2 to Agent~1 is the same in both
  systems, it follows that the instantaneous rate of lost load in the
  $\sim$ system exceeds that in the original system at all times. This
  implies~\eqref{eq:lol_comparison}. It is important to note that a
  similar argument \emph{does not} hold for Agent~2, since it enjoys a
  higher rate of energy transfer from Agent~1 in the $\sim$ system.


  A similar line of reasoning can also be used to prove the
  inequalities for cumulative energy lost due to
  overflow~\eqref{eq:overflow_comparison}. Since
  $\tilde{b}_i(t) \leq b_i(t)$ (Lemma~\ref{lem:bat_induction}) and
  given that the batteries are driven by the same net generation
  process across both systems, an overflow out of Battery~$i$ in the
  $\sim$ system at any time implies an overflow at at least the same
  rate out of Battery~$i$ in the original system. \vd{why at least??
    It is the same rate.} \jk{Will discuss.}
\end{proof}


We are now ready to give the proof of Theorem~\ref{thm:monotonocity}.

\begin{proof}[Proof of Theorem~\ref{thm:monotonocity}]
  From Lemma~\ref{lem:LOL_overflow}, it follows that
  \begin{align*}
    \lim_{t \rightarrow \infty} \frac {\tilde{\ell}_1(t)}{t}  \geq \lim_{t \rightarrow \infty} \frac{\ell_1(t)}{t} \quad 
    \Rightarrow \widetilde{\LLR}_1 \geq \LLR_1. 
  \end{align*}
  That the latter inequality is \emph{strict} follows from a
  straightforward renewal reward argument, which we sketch here. Given
  our coupling between the original and~$\sim$ systems, consider a
  renewal process, where the renewal instants correspond to hitting
  times of the configuration $X(t) = s_2$ and $b_i(t) = 0$ for
  $i = 1,2.$ Let $\ell_i(n)$ (respectively, $\tilde{\ell}_i(n)$)
  denote the lost load in the $n$th renewal cycle by Agent~$i$ in the
  original (respectively, the $\sim$) system. Thus, by the renewal
  reward theorem,
  $$\LLR_i = \frac{\Exp{\ell_i(1)}}{\Exp{T}}, \quad \widetilde{\LLR}_i =
  \frac{\Exp{\tilde{\ell}_i(1)}}{\Exp{T}},$$ where $T$ denotes the
  length of a typical renewal cycle. Using the above characterization,
  it suffices to show that
  $\Exp{\tilde{\ell}_i(1)} > \Exp{\ell_i(1)}.$ Since
  $\tilde{\ell}_i(1) \geq \ell_i(1)$ on all sample paths, one has to
  simply argue that with positive probability,
  $\tilde{\ell}_i(1) > \ell_i(1).$ This is not hard to show.

  Thus, we have
  \begin{equation}
    \label{eq:LLR1_inequality}
    \widetilde{\LLR}_1 > \LLR_1.
  \end{equation}
  It therefore follows more generally that $\LLR_i(c_1,c_2)$ is
  strictly increasing in $c_i$.

  Next, recall that under our coupling of the background process for
  the original system and the $\sim$ system, the total energy received
  in original system within any time interval $[0,\ t]$ is equal to
  the total energy received in $\sim$ system within the same time
  interval. Thus,
  \begin{multline}\label{eq:energy_cons}
    b_1(t)+b_2(t)+\ell_1^c(t)+\ell_2^c(t)+ \mathcal{O}_1(t) + \mathcal{O}_2(t)\\
    = \tilde{b}_1(t)+\tilde{b}_2(t)+\tilde{\ell}_1^c(t)+\tilde{\ell}_2^c(t)+ \tilde{\mathcal{O}}_1(t) + \tilde{\mathcal{O}}_2(t),
  \end{multline}
  where $\ell_i^c(t)$ (respectively, $\tilde{\ell}_i^c(t)$) is the
  cumulative load catered (i.e., demand supplied) by Agent~$i$ over
  the interval $[0\ t]$ in the original system (respectively, the
  $\sim$ system). Since
  ${\mathcal{O}}_i(t) \geq \tilde{{\mathcal{O}}}_i(t)$ (by
  Lemma~\ref{lem:LOL_overflow}),
  \begin{equation}\label{eq:tot_overflow_relation}
    {\mathcal{O}}_1(t) + {\mathcal{O}}_2(t) \geq  \tilde{{\mathcal{O}}}_1(t) + \tilde{{\mathcal{O}}}_2(t).
  \end{equation} 
  Similarly, since $ b_i(t) \geq \tilde{b}_i(t)$ (by
  Lemma~\ref{lem:bat_induction}), we have
  \begin{equation}\label{eq:tot_bat_relation}
    b_1(t)+b_2(t) \geq \tilde{b}_1(t)+\tilde{b}_2(t).
  \end{equation}
  Therefore, using \eqref{eq:tot_overflow_relation} and
  \eqref{eq:tot_bat_relation} in \eqref{eq:energy_cons}, we get
$$\ell_1^c(t)+\ell_2^c(t) \leq \tilde{\ell}_1^c(t)+\tilde{\ell}_2^c(t)$$
i.e., the total cumulative load served or demand supplied over $[0,t]$
in the original system is less than or equal to that in the $\sim$
system. It follows now that within the interval $[0,t]$, the
cumulative lost load in original system is greater than or equal to
that in $\sim$ system, i.e.,
$\tilde{\ell}_1(t)+\tilde{\ell}_2(t) \leq \ell_1(t) +\ell_2(t),$ which implies
\begin{align}
  \lim_{t \rightarrow \infty} \bigg\{\frac{\tilde{\ell}_1(t)}{t}  +  \frac{\tilde{\ell}_2(t)}{t} \bigg\} &\leq \lim_{t \rightarrow \infty} \bigg\{\frac{\ell_1(t)}{t} +\frac{\ell_2(t)}{t} \bigg\} \nonumber \\
  \imp \widetilde{\LLR}_1 +\widetilde{\LLR}_2 &\leq \LLR_1 +\LLR_2. \nonumber
\end{align}
Again, it can be shown that the above inequality is strict via a
renewal reward argument, so that
\begin{equation}
  \label{eq:tot_LLR_inequality}
  \widetilde{\LLR}_1 +\widetilde{\LLR}_2 < \LLR_1 +\LLR_2.
\end{equation}
This shows that more generally, $\LLR_1(c_1,c_2) + \LLR_1(c_1,c_2)$ is
strictly decreasing with respect to any $c_i.$

Finally, it follows from \eqref{eq:LLR1_inequality} and
\eqref{eq:tot_LLR_inequality} that 
 $ \widetilde{\LLR}_2 < \LLR_2,$
which establishes that $\LLR_i(c_1,c_2)$ is strictly decreasing in
$c_{-i}$.
\end{proof}

Armed with Theorem~\ref{thm:monotonocity}, we address the existence of
mutually beneficial sharing configurations, and the structure of the
Pareto frontier of efficient sharing configurations in the following
section. We conclude this section by noting that $\LLR_i(\cdot,\cdot)$
is a continuous function, specifically Lipschitz continuous.
\begin{lemma}
  \label{lemma:Lipschitz}
  For $i \in \{1,2\},$
  $|\LLR_i(c_1 + \epsilon,c_2) - \LLR_i(c_1,c_2)| \leq \epsilon.$
\end{lemma}
\begin{proof}
  Using the same notation as before, it is not hard to see that
  $\tilde{\ell}_1(t) \leq \ell_1(t) +\epsilon t;$ indeed, occasionally
  transferring energy at an additional rate of $\epsilon$ to Agent~2
  results in Agent~1 losing at most $\epsilon t$ additional load until
  time~$t.$ This implies
  $\LLR_1(c_1 + \epsilon,c_2) - \LLR_1(c_1,c_2) \leq \epsilon.$ A
  similar argument also applies to $\LLR_2.$
\end{proof}
An immediate consequence of Lemma~\ref{lemma:Lipschitz} is that
$\LLR_i(\cdot,\cdot)$ is differentiable almost everywhere.

\section{Pareto-optimal sharing}
\label{sec:pareto}

In this section, our goal is to shed light on meaningful sharing
configurations that the agents might agree upon from a game theoretic
standpoint. Our first step is to show that there exist sharing
configurations that are beneficial to both agents.

\subsection{Existence of mutually beneficial configurations}
\label{subsec:existence_of_beneficial_config}

Recall that $c_i$ denotes the maximum rate at which Agent~$i$ shares
energy with Agent~$-i$ when the latter faces loss of load. Moreover,
$c_i \leq c_{i,\max};$ this constraint incorporates any physical
transmission constraints that limit the flow of energy from one agent
to the other. Thus, the space of possible sharing configurations is
given by
$$\mathcal{X}=[0,\ c_{1,\max}] \times [0,\ c_{2,\max}].$$ We also
define $$\mathcal{X}^o=[0,\ c_{1,\max})\times [0,\ c_{2,\max}).$$ Note
that $\mathcal{X}^o$ excludes those configurations where either or
both of the parameters $c_i$ take their maximum value.

Our first result shows that given any configuration in
$\mathcal{X}^o,$ it is possible to perturb it in such a manner that
\emph{both agents are comparatively better off}. Specializing this
result to the configuration $(0,0),$ we conclude that mutually
beneficial configurations are guaranteed to exist.\footnote{Recall
  that the configuration $(0,0)$ is itself mutually beneficial as
  compared to standalone operation (i.e., without overflow sharing).}
Remarkably, this is true no matter how asymmetric the net generation
processes might be across the agents; it is \emph{always} possible to
find sharing arrangements that enable the agents to help one another.



\begin{lemma}{\label{lem:grad_LLR}}
  For any sharing configuration $(c_1,c_2) \in \mathcal{X}^o$, there
  exists a direction $(1,\theta)$, where $\theta >0$, such that the
  gradients $\nabla \LLR_i(c_1,c_2)$ satisfy
  $ \nabla \LLR_i(c_1,c_2) \cdot (1,\theta) <0 \ \forall \ i \in \{1,2\}.$
\end{lemma}


\begin{proof}[Proof of Lemma~\ref{lem:grad_LLR}]
  Since the total loss of load rate $\LLR_1(c_1,c_2)+\LLR_2(c_1,c_2)$
  strictly decreases with $c_i$ (see Theorem~\ref{thm:monotonocity}),
  differentiating with respect to $c_1$ and $c_2$, we get
  \begin{align*}
    &\frac{\partial \LLR_1}{\partial c_1} < -\frac{\partial \LLR_2}{\partial c_1}, \quad
    \frac{\partial \LLR_2}{\partial c_2} < -\frac{\partial \LLR_1}{\partial c_2}. \\
    \Rightarrow  & \frac{\partial \LLR_1}{\partial c_1} \frac{\partial \LLR_2}{\partial c_2}  < (-\frac{\partial \LLR_2}{\partial c_1}) (-\frac{\partial \LLR_1}{\partial c_2} )
  \imp  \displaystyle \frac{\frac{\partial \LLR_1}{\partial c_1}}{-\frac{\partial \LLR_1}{\partial c_2}} < \frac{ -\frac{\partial \LLR_2}{\partial c_1}}{\frac{\partial \LLR_2}{\partial c_2}}
  \end{align*}
  Therefore, there exists a $\theta>0$ such that 
  \begin{equation*}
    \frac{\frac{\partial \LLR_1}{\partial c_1}}{-\frac{\partial \LLR_1}{\partial c_2}} < \theta < \frac{ -\frac{\partial \LLR_2}{\partial c_1}}{\frac{\partial \LLR_2}{\partial c_2}},
  \end{equation*}
  which implies that $ \nabla \LLR_1 \cdot (1,\theta) <0$
  and $ \nabla \LLR_2 \cdot (1,\theta) <0$.
\end{proof}

Next, we turn to the set of Pareto-optimal sharing configurations.

\subsection{Pareto frontier}

We begin by defining Pareto-optimal configurations.
\begin{definition}
  An energy sharing configuration $(c_1,c_2) \in \X$ is Pareto-optimal
  if
  there does not exist $(\hat{c}_1,\hat{c}_2) \in \X$ such that
  $\LLR_i(\hat{c}_1,\hat{c}_2) \leq \LLR_i(c_1,c_2)$ for all
  $i \in \{1,2\},$ the inequality being strict for at least one $i.$
\end{definition}
The set of Pareto-optimal sharing configurations is denoted
by~$\mathcal{P},$ and is referred to as the Pareto
frontier. Pareto-optimal sharing configurations are \emph{efficient},
in the sense that there do not exist configurations that dominate
them. In other words, over the Pareto frontier, the $\LLR$ of any
agent can be only improved by increasing the $\LLR$ of the other.

Lemma~\ref{lem:grad_LLR} shows that there are no Pareto-optimal
configurations in $\mathcal{X}^o.$ The following theorem shows that
the Pareto frontier is in fact $\X \setminus \X^o.$ Moreover, the
Pareto frontier contains mutually beneficial configurations.

\begin{theorem}
\label{thm:Pareto_optimal_sharing}
The Pareto frontier $\mathcal{P} = \X \setminus \X^o.$ Moreover, there
exists points $(c_1,c_2) \in \mathcal{P}$ satisfying
$\LLR_i(c_1,c_2) < \LLR^{sa}_i$ for all $i \in \{1,2\}.$
\end{theorem}

Theorem~\ref{thm:Pareto_optimal_sharing} shows that all efficient
sharing configurations involve at least one agent transferring energy
to the other at the maximum possible rate when the latter faces loss
of load. Moreover, there exist efficient configurations that are
mutually beneficial relative to standalone operation. Intuitively, if
Agent~$i$ is considerably more deficit prone than Agent~$-i,$ then
efficient, mutually beneficial would involve $c_i = c_{i,\max}.$ This
way, Agent~$i$ transfers energy to Agent~$-i$ at the maximum possible
rate at those relatively rare instances when the latter faces loss of
load load. In return, Agent~$-i$ agrees to a modest rate of energy
transfer to Agent~$i$ at those (relatively more often) times when it
faces deficit.


\begin{figure*}[t]
	\begin{minipage}{.33\textwidth}
		\includegraphics[width=\linewidth]{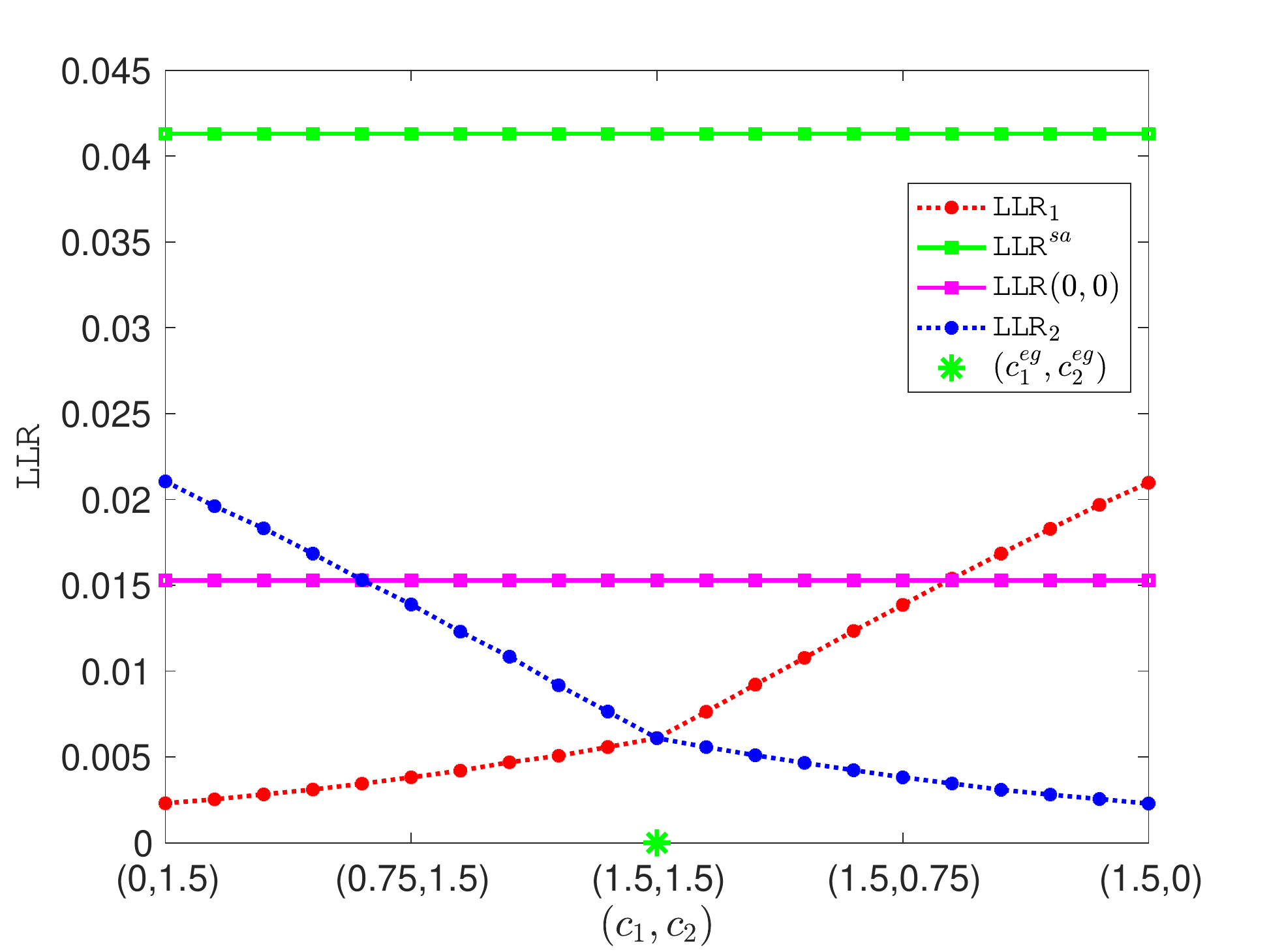}
		\subcaption{$\boldsymbol{r_1(t),r_2(t)\in\{-1.5,2\}}$}
	\end{minipage}
	\begin{minipage}{.33\textwidth}
		\includegraphics[width=\linewidth]{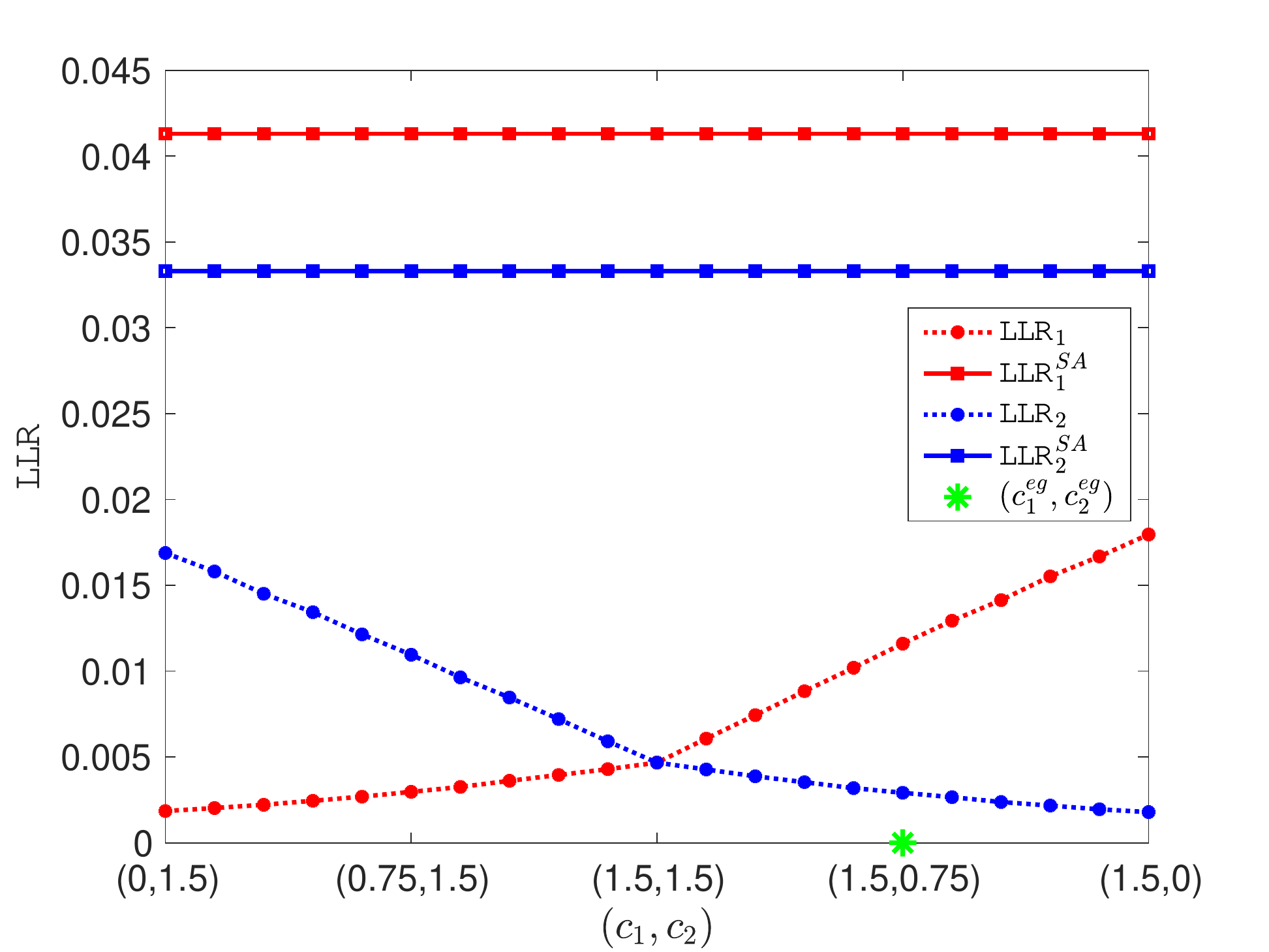}
		\subcaption{$\boldsymbol{r_1(t)\in\{-1.5,2\}}$ and $\boldsymbol{r_2(t)\in\{-1.5,2.15\}}$}
	\end{minipage}
	\begin{minipage}{.33\textwidth}
		\includegraphics[width=\linewidth]{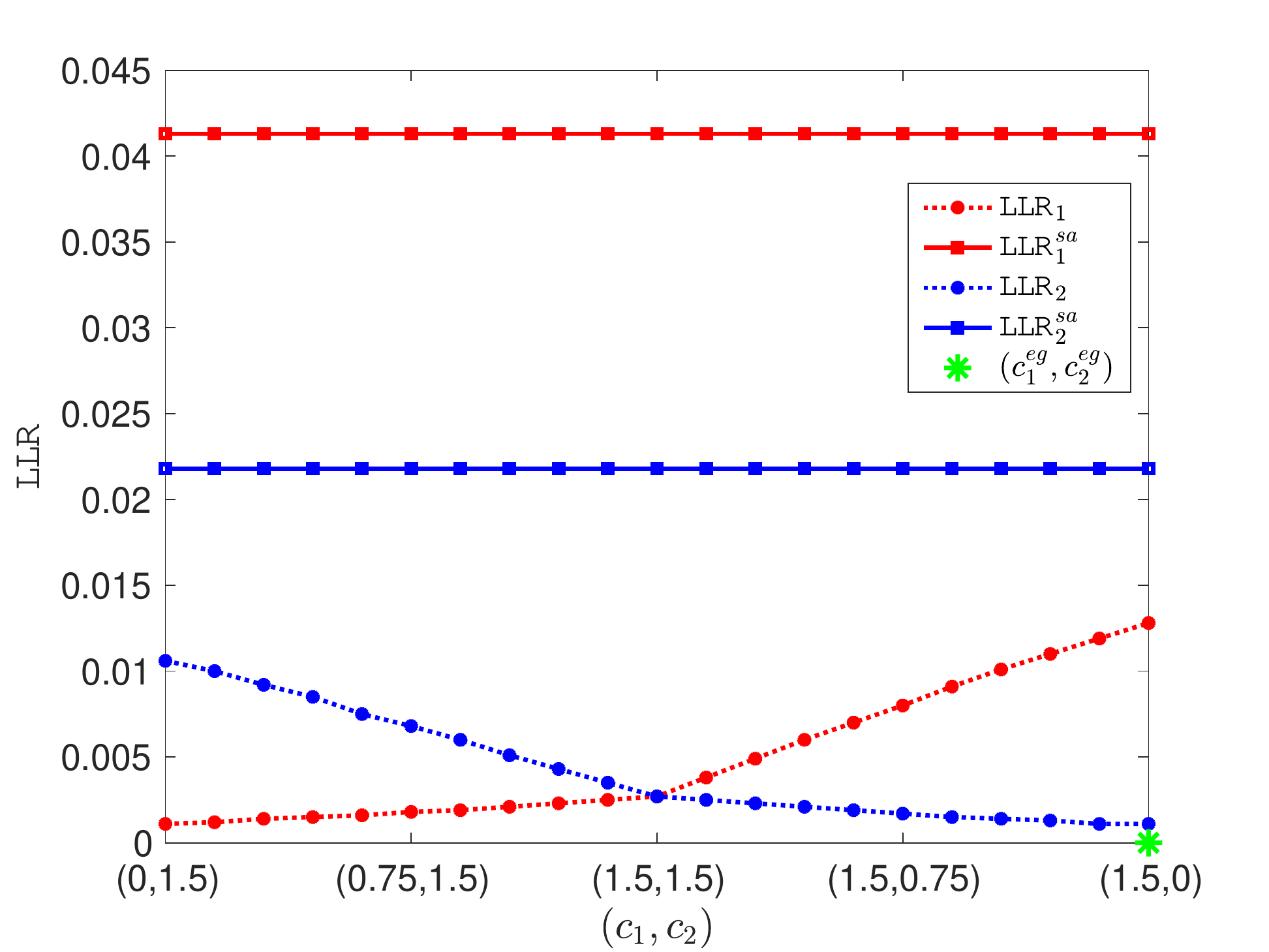}
		\subcaption{$\boldsymbol{r_1(t)\in\{-1.5,2\}}$ and $\boldsymbol{r_2(t)\in\{-1.5,2.5\}}$}
	\end{minipage}
	\caption{Toy example all case: $\boldsymbol{\LLR}$
		variation with sharing configuration
		$\boldsymbol{(c_1,c_2) \in \mathcal{X} \setminus
			\mathcal{X}^o}$ when the net generation processes of two
		agents are symmetric(left) and asymmetric (middle, right). $\boldsymbol{B_1=B_2=10}$ }
	\label{fig:toy_example}
\end{figure*}


\begin{proof}[Proof of Theorem~\ref{thm:Pareto_optimal_sharing}]
  That the Pareto frontier consists of all points in
  $\X \setminus \X^o$ follows by noting that these configurations
  cannot be dominated by any other configuration in $\X.$ To see this,
  suppose that $(c_1,c_2) \in \X \setminus \X^o.$ From
  Theorem~\ref{thm:monotonocity}, it follows that no other
  configuration in $\X \setminus \X^o$ dominates $(c_1,c_2).$ For the
  purpose of obtaining a contradiction, suppose that there exists
  $(c_1',c_2') \in \X^o$ that dominates $(c_1,c_2).$ But
  Lemma~\ref{lem:grad_LLR} shows that there exists
  $(c_1'',c_2'') \in \X \setminus \X^o$ that further dominates
  $(c_1',c_2'),$ and thus also dominates $(c_1,c_2),$ yielding a
  contradiction.

  To argue that there exist mutually beneficial configurations in the
  Pareto frontier, consider the optimization
  $$\max_{(c_1,c_2) \in \X} [\LLR_1(0,0) - \LLR_1(c_1,c_2)]_+ 
  [\LLR_2(0,0) - \LLR_2(c_1,c_2)]_+.$$ Since this is the maximization
  of a continuous function over a compact set, an optimal solution,
  say $(c_1^*,c_2^*),$ exists. Moreover, in light of
  Lemma~\ref{lem:grad_LLR}, the objective value at $(c_1^*,c_2^*)$ is
  strictly positive. This means $(c_1^*,c_2^*)$ is Pareto-optimal, and
  also satisfies $\LLR_i(c_1,c_2) < \LLR_i(0,0) < \LLR_i^{sa}$ for
  both $i \in \{1,2\}.$
\end{proof}

\subsection{Bargaining solutions}

Having proved that there exist configurations on the Pareto frontier
that benefit both agents relative to standalone operation, the next
natural question is to capture the configuration that the (strategic)
agents would agree upon. Of course, the set of mutually beneficial
configurations on the Pareto frontier involves a tradeoff between the
payoffs of the two agents, so the agreement would have to balance the
gains of the two agents. This is precisely the question that the
theory of bargaining addresses \cite{Myerson1991}.

Bargaining theory proposes various solution concepts that seek to
capture the agreement that selfish agents would agree upon, including
the celebrated Nash bargaining solution, the Kalai-Smorodinsky
solution, the egalitarian solution and the utilitarian solution; each
of these solution concepts have an elegant axiomatic justification
\cite{Myerson1991}. For simplicity, we restrict attention to the
egalitarian solution in this paper, which tries to balance the
benefits of the agents as far as possible.

Formally, the egalitarian solution $(c_1^{eg},c_2^{eg})$ is defined
by
$$(c_1^{eg},c_2^{eg}) = \argmax_{(c_1,c_2) \in \mathcal{P}}
\min_{i \in \{1,2\}} \{[\LLR_i^{sa} - \LLR_i(c_1,c_2)]_+\}.$$ The
egalitarian solution is \emph{fair}, in that it maximizes the minimum
benefit (measured by the reduction in $\LLR$ relative to standalone
operation) of the agents. If feasible, the egalitarian solution would
in fact equalize these benefits. In the following section, we
illustrate the benefits that can be achieved in practice with our
proposed sharing mechanism under the egalitarian solution.

\section{Case study}
\label{sec:case_study}

To validate our theoretical findings, in this section, we present a
simulation study using a toy model comprising two independent Markov
net generation processes, and a case study involving a solar generator
and a wind generator using real-world data traces.

\subsection{Toy example}

We consider the net generation processes of the two agents as
independent, two-state CTMCs.
Specifically, the rate matrix corresponding to each agent is taken as
$Q=\begin{bmatrix}
  -1 & 1\\
  1 &-1
\end{bmatrix}$. This means that for each agent, state transitions
occur after exponentially distributed intervals of unit mean
length. We consider three cases for the net generation values:
\begin{itemize}
\item Symmetric case: $r_1(t),r_2(t)\in\{-1.5,2\}$
\item Asymmetric case 1: $r_1(t)\in\{-1.5,2\}$ and $r_2(t)\in\{-1.5,2.15\}$
\item Asymmetric case 2: $r_1(t)\in\{-1.5,2\}$ and $r_2(t)\in\{-1.5,2.5\}$ 
\end{itemize}
Top to bottom, note that Agent~2 becomes increasingly `generative'
relative to Agent~1 in the above scenarios. The battery sizes are set
as $B_1=B_2=10$. For each sharing configuration $(c_1,c_2)$, the CTMCs
$\{r_1(t)\}$ and $\{r_2(t)\}$ are simulated for a suitably long
duration and the temporal evolution of battery and the incurred loss
of load for each agent is captured. The loss of load rate $(\LLR_i)$
for Agent~$i$ is computed as the cumulative loss of load over the
simulation horizon, divided by the horizon
length. Figure~\ref{fig:toy_example} depicts the variation of the
$\LLR$ of each agent over the Pareto frontier for all three
settings. For convenience of presentation, we have `flattened' the
Pareto frontier to be the horizontal axis in these figures; i.e., the
independent variable ranges from $(0,c_{2\max})$ to $(c_{1\max},0)$
via $(c_{1\max},c_{2\max}),$ covering entire range of the Pareto
frontier $\mathcal{X} \setminus \mathcal{X}^o.$ The constraint
parameter $c$ is set as 1.5 (and hence it only restricts the
overflow), so that $c_{1,\max} = c_{2,\max} = 1.5.$

The monotonicity properties of $\LLR$ (Theorem \ref{thm:monotonocity})
are evident from Figure~\ref{fig:toy_example}. For the symmetric case
(panel~$(a)$), each agent experiences the same standalone
$\LLR_i^{sa}$ and the same value of ${\LLR}_i{(0,0)}$. For the
asymmetric cases, the standalone loss of load rates of the two agents
$\LLR^{sa}$ are different (see panels~$(b)$ and~$(c)$). As expected,
the egalitarian bargaining solution for the symmetric case is found to
be full sharing, i.e., $(1.5, 1.5),$ resulting in an 85\% reduction in
$\LLR$ for each agent relative to the standalone setting. In
asymmetric case~1, where Agent~2 becomes more generative, the
bargaining solution shifts `right' to $(1.5,0.75),$ i.e., Agent~2
reduces its peak energy transfer rate to Agent~1. In asymmetric
case~2, where Agent~2 becomes even more generative, the egalitarian
solution shifts further `right' to $(1.5,0),$ i.e., Agent~1 shares
energy with Agent~2 at the maximum rate when the latter faces loss of
load, but Agent~2 only shares its overflow energy with
Agent~1. Intuitively, this is because Agent~2 is substantially less
likely to face loss of load in this example, and substantially more
likely to have an energy overflow. Thus, Agent~1 obtains a
considerable benefit from just receiving this overflow energy, which
it reciprocates by sharing energy with Agent~1 at the maximum rate
when the latter faces a deficit.

\subsection{Energy sharing between wind and solar generator}

We collected six years of wind generation data corresponding to a
location in New Mexico, USA. The data is obtained from the Wind
Integration National Dataset (WIND) Toolkit, which has been made
public by National Renewable Energy Laboratory (NREL)
\cite{url_wind_data_windtoolkit}. The data obtained is of five-minute
temporal resolution and ranges from 0-16 megawatt of wind power. For
the same location, we obtained the hourly solar generation data trace
from the software tool SAM, also available on the NREL
website. Assuming the solar power to be constant over each hour, this
hourly data trace of solar generation, which ranges from 0-16 megawatt
of power, is converted into a time series with a finer temporal
resolution of five minutes.

We consider Agent~1 to be a wind generator, and Agent~2 to be a solar
generator, with the above generation traces. To capture demand, 
for Agent~1, we fix a constant demand that is 90\% of the time average
wind generation. Since solar generation is only available from around
7 am to 5 pm (total 10 hours), we chose a demand curve which is
non-zero only during this interval i.e., 7 am to 5 pm. The demand
value during this interval is chosen to be 90 percent of the average
solar generation, where the averaging is only performed over the same
interval.
Each agent is equipped with a battery capacity 500 kWh.

Over the six years of time series data, the standalone $\LLR^{sa}$ of
Agent~1 (the wind generator) is computed to be $1.8873$ MW, whereas
that of Agent~2 (the solar generator) equals $0.7218$ MW. The
considerably higher standalone $\LLR$ for the wind generator suggests
that the wind generation is much more `variable' than the solar
generation.  The loss of load rate of both agents over the (flattened,
as before) Pareto frontier is plotted in
Figure~\ref{fig:LLR_wind_solar}.  Note that there is a substantial
reduction in $\LLR$ (relative to standalone setting) for the wind generator, but only a modest
reduction for the solar generator (again, consistent with the
considerably higher variability of wind generation). Thus, the
egalitarian solution corresponds to
$(c_1^{eg},c_2^{eg})=(c_{1\max},0),$ i.e., the wind generator shares
energy with the solar generator at the peak rate when the latter faces
a deficit, whereas the solar generator only shares its overflow with
the wind generator. The associated reduction in loss of load rate
(compared to the standalone setting) is 70\% for the wind generator
and 22\% for the solar generator.


\begin{figure}
	\includegraphics[scale=0.35]{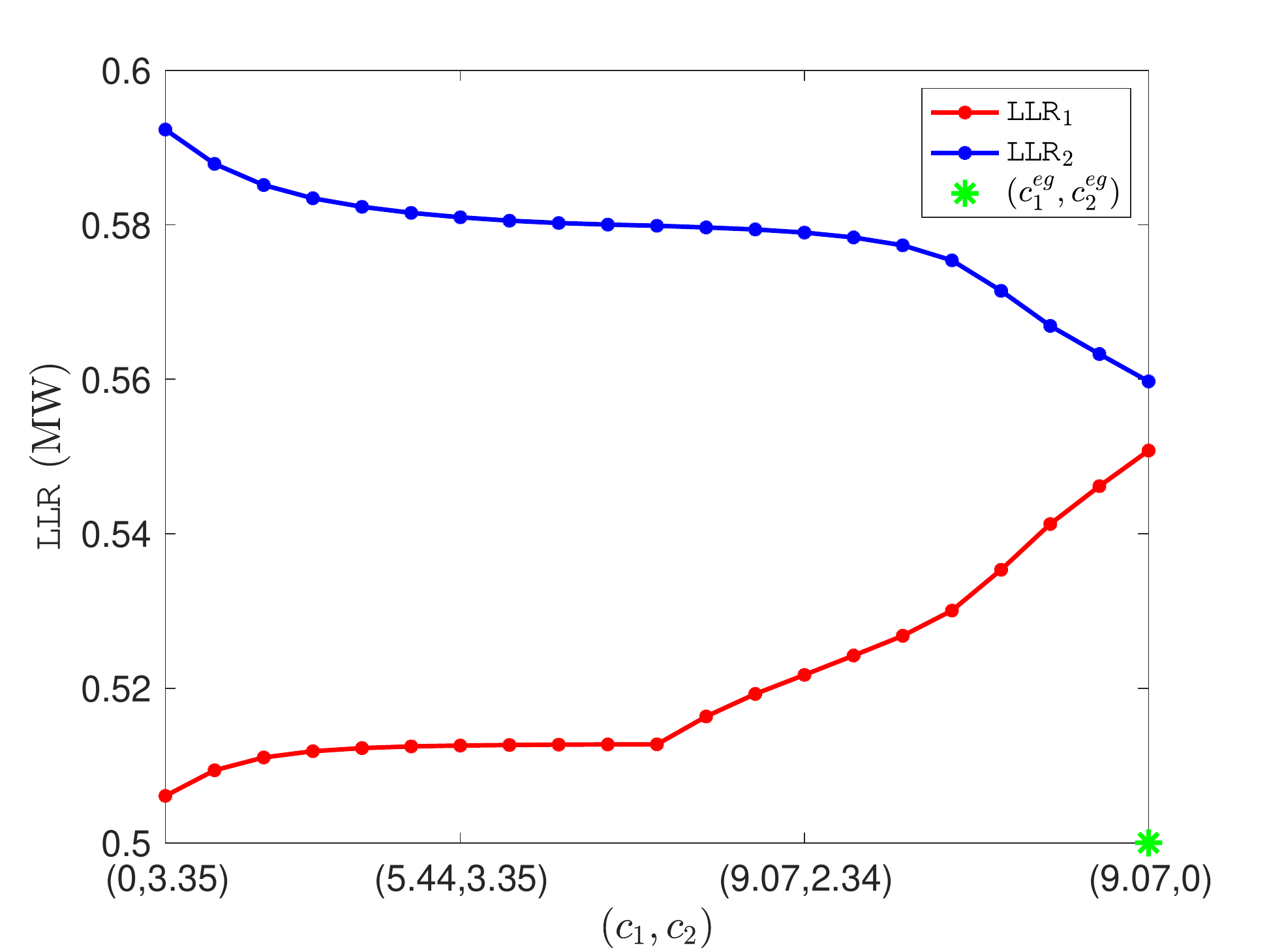}
	\caption{Simulation results with wind and solar data traces. $\boldsymbol{\LLR_1}$ and $\boldsymbol{LLR_2}$ are plotted for various sharing configurations covering the entire Pareto-frontier.  $\boldsymbol{B_1=B_2=500kWh}$}
	\label{fig:LLR_wind_solar}
\end{figure}

\section{Concluding remarks}
\label{sec:discuss}

This work motivates extensions along various directions. A natural
first step would be to generalize the sharing mechanism to multiple
(more than two) agents, and characterize the Pareto frontier of the
achievable reliability vectors. Another useful direction is to
consider other reliability metrics, potentially capturing (i) shorter
timescales, and (ii) a non-linear cost/penalty associated with
increasing loss of load.

Another interesting line of questions pertains to network formation:
Given a collection of renewable generators, which of them should come
together to enter into an energy sharing agreement? Since statistical
diversity lies at the core of effective energy sharing, it is also
important to consider policy interventions that would encourage energy
sharing agreements between geographically separated renewable
generators. Indeed, such generators would have to be allowed to fulfill
a deficit in generation at one bus on the grid with a surplus
injection at another (of course subject to grid stability
considerations).

\section*{Acknowledgment}

This work was supported by grant DST/CERI/MI/SG/2017/077 under the
Mission Innovation program on Smart Grids by the Department of Science
and Technology, India.

\bibliographystyle{IEEEtran}
\bibliography{refs}


\begin{thebibliography}{24}


\ifx \showCODEN    \undefined \def \showCODEN     #1{\unskip}     \fi
\ifx \showDOI      \undefined \def \showDOI       #1{#1}\fi
\ifx \showISBNx    \undefined \def \showISBNx     #1{\unskip}     \fi
\ifx \showISBNxiii \undefined \def \showISBNxiii  #1{\unskip}     \fi
\ifx \showISSN     \undefined \def \showISSN      #1{\unskip}     \fi
\ifx \showLCCN     \undefined \def \showLCCN      #1{\unskip}     \fi
\ifx \shownote     \undefined \def \shownote      #1{#1}          \fi
\ifx \showarticletitle \undefined \def \showarticletitle #1{#1}   \fi
\ifx \showURL      \undefined \def \showURL       {\relax}        \fi
\providecommand\bibfield[2]{#2}
\providecommand\bibinfo[2]{#2}
\providecommand\natexlab[1]{#1}
\providecommand\showeprint[2][]{arXiv:#2}

\bibitem[\protect\citeauthoryear{Anick, Mitra, and Sondhi}{Anick
  et~al\mbox{.}}{1982}]%
        {anick1982stochastic}
\bibfield{author}{\bibinfo{person}{David Anick}, \bibinfo{person}{Debasis
  Mitra}, {and} \bibinfo{person}{Man~Mohan Sondhi}.}
  \bibinfo{year}{1982}\natexlab{}.
\newblock \showarticletitle{Stochastic theory of a data-handling system with
  multiple sources}.
\newblock \bibinfo{journal}{\emph{Bell System Technical Journal}}
  \bibinfo{volume}{61}, \bibinfo{number}{8} (\bibinfo{year}{1982}),
  \bibinfo{pages}{1871--1894}.
\newblock


\bibitem[\protect\citeauthoryear{Archer, Malita, Moraes, and Plotky}{Archer
  et~al\mbox{.}}{2006}]%
        {archer2006load}
\bibfield{author}{\bibinfo{person}{Shafford Archer}, \bibinfo{person}{Florin
  Malita}, \bibinfo{person}{Ian Moraes}, {and} \bibinfo{person}{Jon Plotky}.}
  \bibinfo{year}{2006}\natexlab{}.
\newblock \bibinfo{title}{Load balancing in a distributed telecommunications
  platform}.
\newblock
\newblock
\newblock
\shownote{US Patent App. 11/170,457.}


\bibitem[\protect\citeauthoryear{Association}{Association}{2015}]%
        {EWEA_energy_imbalance_penalty}
\bibfield{author}{\bibinfo{person}{The European Wind~Energy Association}.}
  \bibinfo{year}{2015}\natexlab{}.
\newblock \bibinfo{booktitle}{\emph{Balancing responsibility and costs of wind
  power plants}}.
\newblock
\urldef\tempurl%
\url{https://www.ewea.org/fileadmin/files/library/publications/position-papers/EWEA-position-paper-balancing-responsibility-and-costs.pdf}
\showURL{%
Retrieved February 8, 2020 from \tempurl}


\bibitem[\protect\citeauthoryear{Brunetto and Tina}{Brunetto and Tina}{2011}]%
        {brunetto2011wind}
\bibfield{author}{\bibinfo{person}{C Brunetto} {and} \bibinfo{person}{G Tina}.}
  \bibinfo{year}{2011}\natexlab{}.
\newblock \showarticletitle{Wind generation imbalances penalties in day-ahead
  energy markets: The Italian case}.
\newblock \bibinfo{journal}{\emph{Electric power systems research}}
  \bibinfo{volume}{81}, \bibinfo{number}{7} (\bibinfo{year}{2011}),
  \bibinfo{pages}{1446--1455}.
\newblock


\bibitem[\protect\citeauthoryear{Cardellini, Colajanni, and Yu}{Cardellini
  et~al\mbox{.}}{1999}]%
        {cardellini1999dynamic}
\bibfield{author}{\bibinfo{person}{Valeria Cardellini},
  \bibinfo{person}{Michele Colajanni}, {and} \bibinfo{person}{Philip~S Yu}.}
  \bibinfo{year}{1999}\natexlab{}.
\newblock \showarticletitle{Dynamic load balancing on web-server systems}.
\newblock \bibinfo{journal}{\emph{IEEE Internet computing}}
  \bibinfo{volume}{3}, \bibinfo{number}{3} (\bibinfo{year}{1999}),
  \bibinfo{pages}{28--39}.
\newblock


\bibitem[\protect\citeauthoryear{Carri{\'o}n, Arroyo, and Conejo}{Carri{\'o}n
  et~al\mbox{.}}{2009}]%
        {carrion2009bilevel}
\bibfield{author}{\bibinfo{person}{Miguel Carri{\'o}n},
  \bibinfo{person}{Jos{\'e}~M Arroyo}, {and} \bibinfo{person}{Antonio~J
  Conejo}.} \bibinfo{year}{2009}\natexlab{}.
\newblock \showarticletitle{A bilevel stochastic programming approach for
  retailer futures market trading}.
\newblock \bibinfo{journal}{\emph{IEEE Transactions on Power Systems}}
  \bibinfo{volume}{24}, \bibinfo{number}{3} (\bibinfo{year}{2009}),
  \bibinfo{pages}{1446--1456}.
\newblock


\bibitem[\protect\citeauthoryear{{Cintuglu}, {Martin}, and
  {Mohammed}}{{Cintuglu} et~al\mbox{.}}{2015}]%
        {Cintuglu15}
\bibfield{author}{\bibinfo{person}{M.~H. {Cintuglu}}, \bibinfo{person}{H.
  {Martin}}, {and} \bibinfo{person}{O.~A. {Mohammed}}.}
  \bibinfo{year}{2015}\natexlab{}.
\newblock \showarticletitle{Real-Time Implementation of Multiagent-Based Game
  Theory Reverse Auction Model for Microgrid Market Operation}.
\newblock \bibinfo{journal}{\emph{IEEE Transactions on Smart Grid}}
  \bibinfo{volume}{6}, \bibinfo{number}{2} (\bibinfo{year}{2015}),
  \bibinfo{pages}{1064--1072}.
\newblock


\bibitem[\protect\citeauthoryear{Dahlin, Wang, Anderson, and Patterson}{Dahlin
  et~al\mbox{.}}{1994}]%
        {dahlin1994cooperative}
\bibfield{author}{\bibinfo{person}{Michael~D Dahlin},
  \bibinfo{person}{Randolph~Y Wang}, \bibinfo{person}{Thomas~E Anderson}, {and}
  \bibinfo{person}{David~A Patterson}.} \bibinfo{year}{1994}\natexlab{}.
\newblock \showarticletitle{Cooperative caching: Using remote client memory to
  improve file system performance}. In \bibinfo{booktitle}{\emph{Proceedings of
  the 1st USENIX conference on Operating Systems Design and Implementation}}.
  \bibinfo{pages}{19--es}.
\newblock


\bibitem[\protect\citeauthoryear{Deulkar, Nair, and Kulkarni}{Deulkar
  et~al\mbox{.}}{2019}]%
        {deulkar2019sizing}
\bibfield{author}{\bibinfo{person}{Vivek Deulkar},
  \bibinfo{person}{Jayakrishnan Nair}, {and} \bibinfo{person}{Ankur~A
  Kulkarni}.} \bibinfo{year}{2019}\natexlab{}.
\newblock \showarticletitle{Sizing Storage for Reliable Renewable Integration}.
  In \bibinfo{booktitle}{\emph{2019 IEEE Milan PowerTech}}. IEEE,
  \bibinfo{pages}{1--6}.
\newblock


\bibitem[\protect\citeauthoryear{Dimeas and Hatziargyriou}{Dimeas and
  Hatziargyriou}{2005}]%
        {dimeas2005operation}
\bibfield{author}{\bibinfo{person}{Aris~L Dimeas} {and}
  \bibinfo{person}{Nikos~D Hatziargyriou}.} \bibinfo{year}{2005}\natexlab{}.
\newblock \showarticletitle{Operation of a multiagent system for microgrid
  control}.
\newblock \bibinfo{journal}{\emph{IEEE Transactions on Power systems}}
  \bibinfo{volume}{20}, \bibinfo{number}{3} (\bibinfo{year}{2005}),
  \bibinfo{pages}{1447--1455}.
\newblock


\bibitem[\protect\citeauthoryear{Eddy, Gooi, and Chen}{Eddy
  et~al\mbox{.}}{2014}]%
        {eddy2014multi}
\bibfield{author}{\bibinfo{person}{YS~Foo Eddy}, \bibinfo{person}{Hoay~Beng
  Gooi}, {and} \bibinfo{person}{Shuai~Xun Chen}.}
  \bibinfo{year}{2014}\natexlab{}.
\newblock \showarticletitle{Multi-agent system for distributed management of
  microgrids}.
\newblock \bibinfo{journal}{\emph{IEEE Transactions on power systems}}
  \bibinfo{volume}{30}, \bibinfo{number}{1} (\bibinfo{year}{2014}),
  \bibinfo{pages}{24--34}.
\newblock


\bibitem[\protect\citeauthoryear{Grokop and Tse}{Grokop and Tse}{2008}]%
        {grokop2008spectrum}
\bibfield{author}{\bibinfo{person}{Leonard Grokop} {and}
  \bibinfo{person}{David~NC Tse}.} \bibinfo{year}{2008}\natexlab{}.
\newblock \showarticletitle{Spectrum sharing between wireless networks}. In
  \bibinfo{booktitle}{\emph{IEEE INFOCOM 2008-The 27th Conference on Computer
  Communications}}. IEEE, \bibinfo{pages}{201--205}.
\newblock


\bibitem[\protect\citeauthoryear{Huang, Berry, and Honig}{Huang
  et~al\mbox{.}}{2006}]%
        {huang2006auction}
\bibfield{author}{\bibinfo{person}{Jianwei Huang}, \bibinfo{person}{Randall~A
  Berry}, {and} \bibinfo{person}{Michael~L Honig}.}
  \bibinfo{year}{2006}\natexlab{}.
\newblock \showarticletitle{Auction-based spectrum sharing}.
\newblock \bibinfo{journal}{\emph{Mobile Networks and Applications}}
  \bibinfo{volume}{11}, \bibinfo{number}{3} (\bibinfo{year}{2006}),
  \bibinfo{pages}{405--408}.
\newblock


\bibitem[\protect\citeauthoryear{Laboratory}{Laboratory}{2019}]%
        {url_wind_data_windtoolkit}
\bibfield{author}{\bibinfo{person}{National Renewable~Energy Laboratory}.}
  \bibinfo{year}{2019}\natexlab{}.
\newblock \bibinfo{title}{{Weather Observations}}.
\newblock \bibinfo{howpublished}{\url{https://maps.nrel.gov/wind-prospector/}}.
\newblock
\newblock
\shownote{[Online; accessed 10-February-2020].}


\bibitem[\protect\citeauthoryear{Lakshminarayana, Quek, and
  Poor}{Lakshminarayana et~al\mbox{.}}{2014}]%
        {lakshminarayana2014cooperation}
\bibfield{author}{\bibinfo{person}{Subhash Lakshminarayana},
  \bibinfo{person}{Tony~QS Quek}, {and} \bibinfo{person}{H~Vincent Poor}.}
  \bibinfo{year}{2014}\natexlab{}.
\newblock \showarticletitle{Cooperation and storage tradeoffs in power grids
  with renewable energy resources}.
\newblock \bibinfo{journal}{\emph{IEEE Journal on Selected Areas in
  Communications}} \bibinfo{volume}{32}, \bibinfo{number}{7}
  (\bibinfo{year}{2014}), \bibinfo{pages}{1386--1397}.
\newblock


\bibitem[\protect\citeauthoryear{Lee, Xiang, Schober, and Wong}{Lee
  et~al\mbox{.}}{2014}]%
        {lee2014direct}
\bibfield{author}{\bibinfo{person}{Woongsup Lee}, \bibinfo{person}{Lin Xiang},
  \bibinfo{person}{Robert Schober}, {and} \bibinfo{person}{Vincent~WS Wong}.}
  \bibinfo{year}{2014}\natexlab{}.
\newblock \showarticletitle{Direct electricity trading in smart grid: A
  coalitional game analysis}.
\newblock \bibinfo{journal}{\emph{IEEE Journal on Selected Areas in
  Communications}} \bibinfo{volume}{32}, \bibinfo{number}{7}
  (\bibinfo{year}{2014}), \bibinfo{pages}{1398--1411}.
\newblock


\bibitem[\protect\citeauthoryear{Liu, Yu, Wang, Li, Ma, and Lei}{Liu
  et~al\mbox{.}}{2017}]%
        {liu2017energy}
\bibfield{author}{\bibinfo{person}{Nian Liu}, \bibinfo{person}{Xinghuo Yu},
  \bibinfo{person}{Cheng Wang}, \bibinfo{person}{Chaojie Li},
  \bibinfo{person}{Li Ma}, {and} \bibinfo{person}{Jinyong Lei}.}
  \bibinfo{year}{2017}\natexlab{}.
\newblock \showarticletitle{Energy-sharing model with price-based demand
  response for microgrids of peer-to-peer prosumers}.
\newblock \bibinfo{journal}{\emph{IEEE Transactions on Power Systems}}
  \bibinfo{volume}{32}, \bibinfo{number}{5} (\bibinfo{year}{2017}),
  \bibinfo{pages}{3569--3583}.
\newblock


\bibitem[\protect\citeauthoryear{Luo, Itaya, Nakamura, and Davis}{Luo
  et~al\mbox{.}}{2014}]%
        {luo2014autonomous}
\bibfield{author}{\bibinfo{person}{Yuan Luo}, \bibinfo{person}{Satoko Itaya},
  \bibinfo{person}{Shin Nakamura}, {and} \bibinfo{person}{Peter Davis}.}
  \bibinfo{year}{2014}\natexlab{}.
\newblock \showarticletitle{Autonomous cooperative energy trading between
  prosumers for microgrid systems}. In \bibinfo{booktitle}{\emph{39th annual
  IEEE conference on local computer networks workshops}}.
  \bibinfo{pages}{693--696}.
\newblock


\bibitem[\protect\citeauthoryear{Mitra}{Mitra}{1988}]%
        {Mitra88}
\bibfield{author}{\bibinfo{person}{Debasis Mitra}.}
  \bibinfo{year}{1988}\natexlab{}.
\newblock \showarticletitle{Stochastic theory of a fluid model of producers and
  consumers coupled by a buffer}.
\newblock \bibinfo{journal}{\emph{Advances in Applied Probability}}
  \bibinfo{volume}{20}, \bibinfo{number}{3} (\bibinfo{year}{1988}),
  \bibinfo{pages}{646--676}.
\newblock


\bibitem[\protect\citeauthoryear{Myerson~Roger}{Myerson~Roger}{1991}]%
        {Myerson1991}
\bibfield{author}{\bibinfo{person}{B Myerson~Roger}.}
  \bibinfo{year}{1991}\natexlab{}.
\newblock \bibinfo{title}{Game theory: {A}nalysis of conflict}.
\newblock
\newblock


\bibitem[\protect\citeauthoryear{Nandigam, Jog, Manjunath, Nair, and
  Prabhu}{Nandigam et~al\mbox{.}}{2019}]%
        {Nandigam19}
\bibfield{author}{\bibinfo{person}{Anvitha Nandigam}, \bibinfo{person}{Suraj
  Jog}, \bibinfo{person}{D Manjunath}, \bibinfo{person}{Jayakrishnan Nair},
  {and} \bibinfo{person}{Balakrishna~J Prabhu}.}
  \bibinfo{year}{2019}\natexlab{}.
\newblock \showarticletitle{Sharing within limits: Partial resource pooling in
  loss systems}.
\newblock \bibinfo{journal}{\emph{IEEE/ACM Transactions on Networking}}
  \bibinfo{volume}{27}, \bibinfo{number}{4} (\bibinfo{year}{2019}),
  \bibinfo{pages}{1305--1318}.
\newblock


\bibitem[\protect\citeauthoryear{Sarkar, Singh, and Kumar}{Sarkar
  et~al\mbox{.}}{2008}]%
        {sarkar2008coalitional}
\bibfield{author}{\bibinfo{person}{Saswati Sarkar},
  \bibinfo{person}{Chandramani Singh}, {and} \bibinfo{person}{Anurag Kumar}.}
  \bibinfo{year}{2008}\natexlab{}.
\newblock \showarticletitle{A coalitional game model for spectrum pooling in
  wireless data access networks}. In \bibinfo{booktitle}{\emph{2008 Information
  Theory and Applications Workshop}}. IEEE, \bibinfo{pages}{310--319}.
\newblock


\bibitem[\protect\citeauthoryear{Singh, Sarkar, Aram, and Kumar}{Singh
  et~al\mbox{.}}{2011}]%
        {singh2011cooperative}
\bibfield{author}{\bibinfo{person}{Chandramani Singh}, \bibinfo{person}{Saswati
  Sarkar}, \bibinfo{person}{Alireza Aram}, {and} \bibinfo{person}{Anurag
  Kumar}.} \bibinfo{year}{2011}\natexlab{}.
\newblock \showarticletitle{Cooperative profit sharing in coalition-based
  resource allocation in wireless networks}.
\newblock \bibinfo{journal}{\emph{IEEE/ACM Transactions on Networking}}
  \bibinfo{volume}{20}, \bibinfo{number}{1} (\bibinfo{year}{2011}),
  \bibinfo{pages}{69--83}.
\newblock


\bibitem[\protect\citeauthoryear{Zhu, Huang, Sharma, Su, Irwin, Mishra,
  Menasche, and Shenoy}{Zhu et~al\mbox{.}}{2013}]%
        {zhu2013sharing}
\bibfield{author}{\bibinfo{person}{Ting Zhu}, \bibinfo{person}{Zhichuan Huang},
  \bibinfo{person}{Ankur Sharma}, \bibinfo{person}{Jikui Su},
  \bibinfo{person}{David Irwin}, \bibinfo{person}{Aditya Mishra},
  \bibinfo{person}{Daniel Menasche}, {and} \bibinfo{person}{Prashant Shenoy}.}
  \bibinfo{year}{2013}\natexlab{}.
\newblock \showarticletitle{Sharing renewable energy in smart microgrids}. In
  \bibinfo{booktitle}{\emph{2013 ACM/IEEE International Conference on
  Cyber-Physical Systems (ICCPS)}}. IEEE, \bibinfo{pages}{219--228}.
\newblock


\end{thebibliography}
\section{Appendix}
\label{appendix}

\begin{proof}[Proof of Lemma \ref{lem:bat_induction}]

  Let $\{t_k\}_{k\in \mathbb{Z}^+}$ denote the transition instances
  when the background process $\{X(t)\}$ changes its state, so that
  the net generation values $r_1(t)$ and/or $r_2(t)$, change at these
  instances.  Therefore, within any time slot $(t_k, t_{k+1}]$, the
  net generation is constant for both agents.

  The proof proceeds via induction. We assume
  Lemma~\ref{lem:bat_induction} holds at time $t_k$. We will verify
  that in each of the possible ways in which the dynamics of the
  original system and the $\sim$ system can evolve,
  Lemma~\ref{lem:bat_induction} holds $\forall t\in (t_k,
  t_{k+1}]$. We define instances
  $e_i, \tilde{e}_i, f_i, \tilde{f}_i \in \big(t_k,\ t_{k+1}\big]$ as
  the instants when
  $b_i(e_i)=0,\ \tilde{b}_i(\tilde{e}_i)=0,\ b_i(f_i)=B_i$ and
  $\tilde{b}_i(\tilde{f}_i)=B_i,$ respectively.
  Occasionally, for simplicity, we refer to the battery (not battery
  level) of Agent~$i$ in the original system (respectively, in the
  $\sim$ system) by $b_i$ (respectively, by $\tilde{b}_i$). Let $r_i$
  denote the net generation rate of Agent~$i$ within the slot
  $(t_k, t_{k+1}].$ There are the following four cases to consider.
\begin{caseof}
  \case { $r_1 >0,\ r_2 >0$} {There is no energy sharing in this case,
    so it is easy to see that the ordering of battery occupancies
    between the two systems continues to hold for $t \in $
    $(t_k, t_{k+1}].$}

  \case{ $r_1<0,\ r_2>0$} {Throughout the interval $(t_k \ t_{k+1}]$,
    there is no energy transfer from Agent~1 to Agent~2 in this case,
    in the original system as well as in the $\sim$ system.
    Therefore the extra $\epsilon$ in the sharing configuration of
    Agent~1 in the $\sim$ system is inconsequential and the induction
    step follows as before.}

  \case{$r_1<0,\ r_2<0$} {
    Due to space constraints, we only consider the case where
    $\tilde{b}_1$ gets fully discharged before $\tilde{b}_2$, i.e.,
    $t_k < \tilde{e}_1 < \tilde{e}_2.$ The complementary case can be
    handled on similar lines. The two possible cases are:
    \begin{enumerate}
    \item $b_1$ discharges before $\tilde{b}_2$ which leads to the
      only possible sequence of events as
      $\tilde{e}_1 <e_1 < \tilde{e}_2 < e_2$ (see
      Figure~\ref{fig:Case3_1}).
      \begin{figure}[H]
        \includegraphics[trim={1cm 0 0 0},scale=0.28]{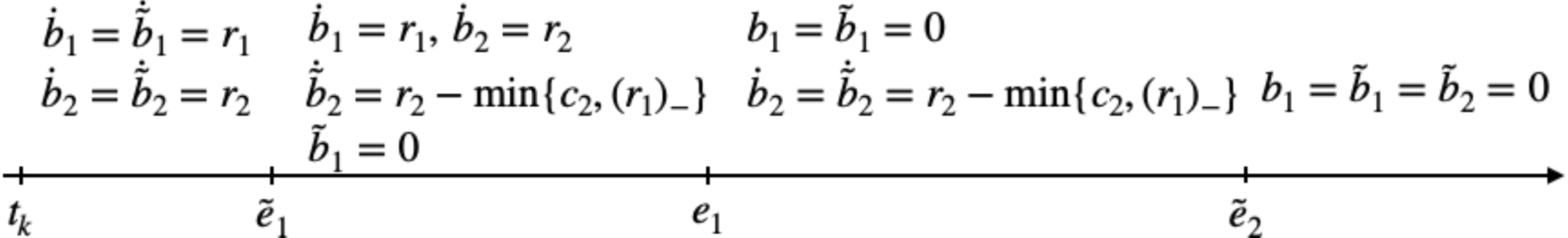}
        \caption{$\boldsymbol{r_1<0,\ r_2<0,\ b_1}$ discharges before $\boldsymbol{\tilde{b}_2}$}	
        \label{fig:Case3_1}
      \end{figure}
      It is clear from Figure~\ref{fig:Case3_1} that the  induction argument in Lemma~\ref{lem:bat_induction} holds $\forall\ t\in [t_k,\ t_{k+1}].$
      
    \item $\tilde{b}_2$ discharges before $b_1$ which leads to the possible sequence of events as shown in Figure~\ref{fig:Case3_2}.
      \begin{figure}[H]
        \includegraphics[trim={0 0 0 0},scale=0.3]{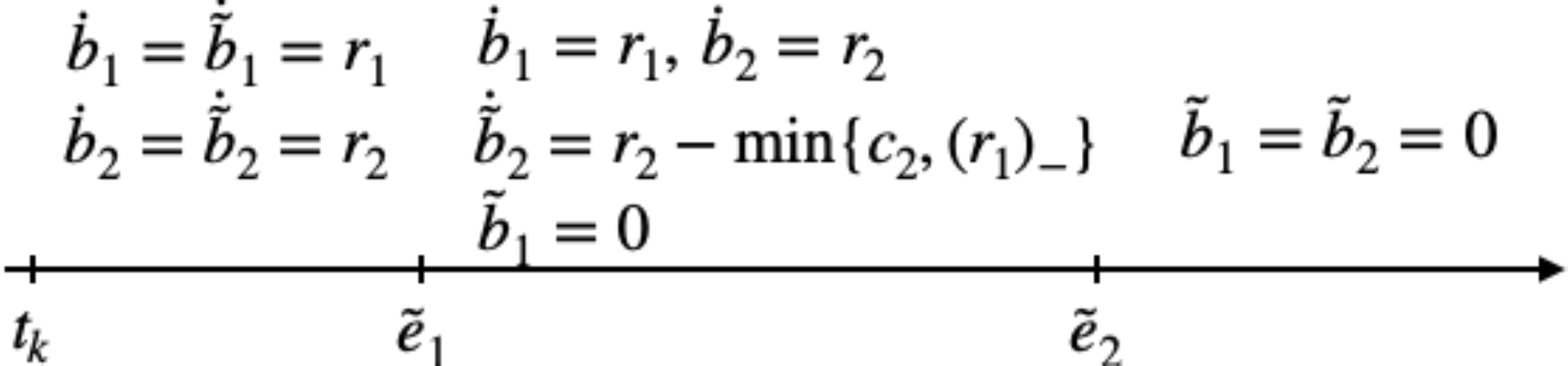}	
        \caption{$\boldsymbol{r_1<0,\ r_2<0,\ \tilde{b}_2}$ discharges before $\boldsymbol{b_1}$}	
        \label{fig:Case3_2}
      \end{figure}
      
      Again, it is clear from Figure~\ref{fig:Case3_2} that
      Lemma~\ref{lem:bat_induction} holds
      $\forall\ t\in [t_k,\ t_{k+1}].$
    \end{enumerate}
  }
  
	 
	
  \case{$r_1 >0,\ r_2 <0$} {In this, there are two major cases to be
    considered, each one needs a separate analysis for
    $r_1 \geq |r_2|$ and for $r_1 < |r_2|$.
    
    \noindent {\bf Subcase 1:} $b_1$ becomes full before $\tilde{b}_2$
    drains to zero ($t_k < f_1 < \tilde{e}_2$).
        
      \noindent {\bf 1a:} $r_1 \geq |r_2|$
        \begin{figure}[H]
          \includegraphics[trim={0 0 0 0},scale=0.4]{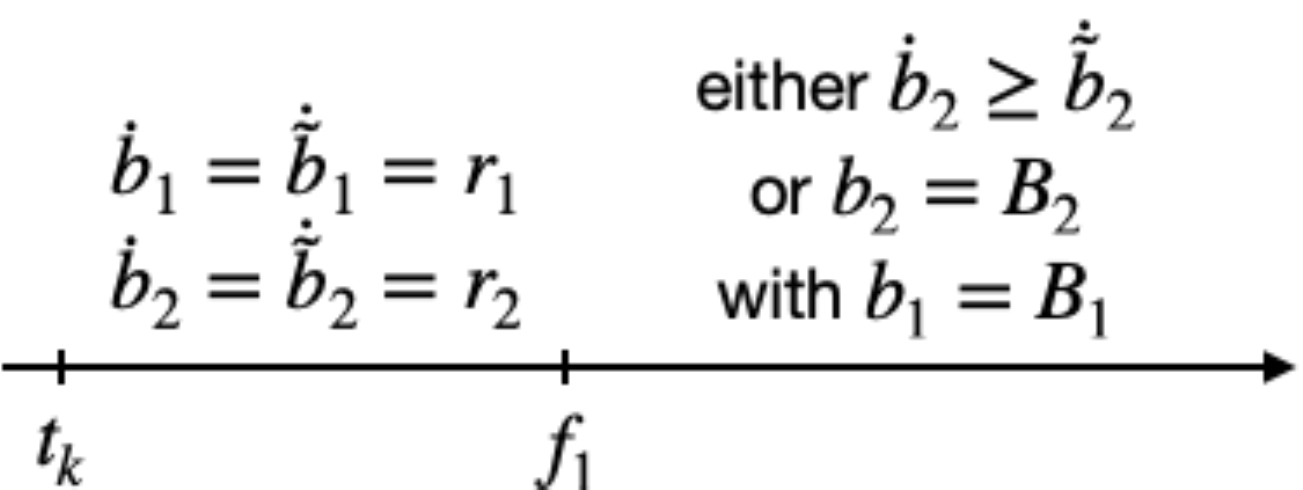}
          \caption{$\boldsymbol{r_2<0<r_1,\ r_1\geq |r_2|}$}
        \end{figure}
        From time $f_1$ onwards, $b_1$ gives energy to $b_2$ with rate
        $r_1$ due to its overflow. Since the overflow is greater than
        or equal to the absolute value of the discharging rate of
        $b_2$, $b_2$ will either charge or remain at the same level
        whereas $\tilde{b}_1$ will continue to charge and won't
        discharge even after $\tilde{b}_2$ discharges to zero. Once
        $\tilde{b}_1$ is full, rate of change of battery levels in
        $b_2$ and $\tilde{b}_2$ is the same (if $b_2$ is not fully
        charged). Therefore from $f_1$ onwards, either $b_2=B_2$ or
        $\dot{b}_2=\dot{\tilde{b}}_2$ which results in
        $b_2 \geq \tilde{b}_2\ \forall t\in [f_1,\ t_{k+1}]$. Also
        note that from $f_1$ onwards, $b_1=B_1$. Therefore together we
        conclude that \cref{lem:bat_induction} holds for
        $\forall t\in [t_k,\ t_{k+1}]$.
 			
      \noindent {\bf 1b:} $r_1 < |r_2|$

      In this there are three further cases to consider:
      
      \noindent {\bf 1b-A:} When full, both $b_1$ and $\tilde{b}_1$
      allows discharge i.e., $r_1<\min(c_1,-r_2)$ and hence
      $r_1<\min(c_1+\epsilon,-r_2).$ Possible event sequences are:
        \begin{enumerate}
        \item $f_1 < \tilde{e}_2 < \tilde{e}_1 	 < e_2			  < e_1$
        \item $f_1 < \tilde{e}_2 < e_2 			  < \tilde{e}_1   < e_1$
        \item $f_1 < \tilde{f}_1  < \tilde{e}_2  < e_2 			  < \tilde{e}_1   < e_1$
        \item $f_1 < \tilde{f}_1  < \tilde{e}_2  < \tilde{e}_1   < e_2            < e_1$
        \end{enumerate}
        It is not hard to verify that, in each case, the battery
        induction argument \eqref{eq:bat_induction} holds along the
        sequence of events. Sequence (4) is depicted in the figure
        below. Other sequences can easily be verified analogously. 
        \begin{figure}[H]
          \includegraphics[trim={0 0 0 0},scale=0.34]{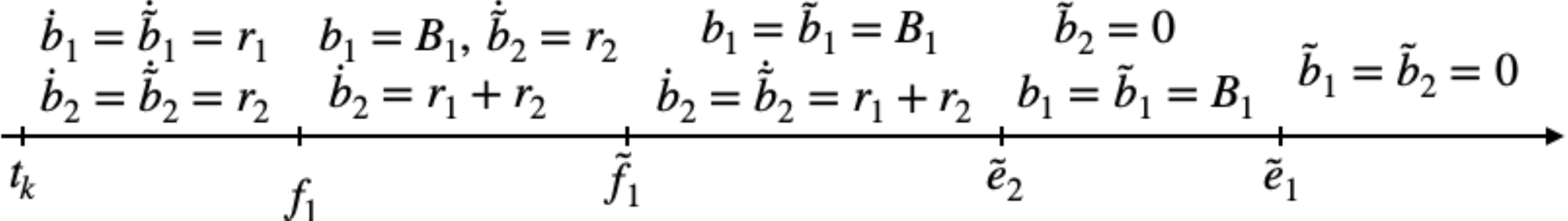}
          \caption{$\boldsymbol{r_2<0<r_1,\ r_1 < |r_2|}$; when full, both $\boldsymbol{b_1}$ and $\boldsymbol{\tilde{b}_1}$ allows discharge}
        \end{figure}
        \noindent {\bf 1b-B:} When full, neither $b_1$ nor
        $\tilde{b}_1$ allows the discharge i.e.
        $r_1 \geq \min(c_1,-r_2)$ and
        $r_1 \geq \min(c_1+\epsilon,-r_2).$ Due to $r_1<|r_2|$,
        possible event sequences are:
        \begin{enumerate}
        \item $f_1 < \tilde{e}_2 < \tilde{f}_1 < e_2$
        \item $f_1 < \tilde{e}_2 < e_2 < \tilde{f}_1$
        \item  $f_1 < \tilde{f}_1  < \tilde{e}_2 < e_2$
        \end{enumerate}
        Sequence (3) is depicted in the figure below. Other sequences can be
        verified similarly.
        \begin{figure}[H]
          \includegraphics[trim={0 0 0 0},scale=0.4]{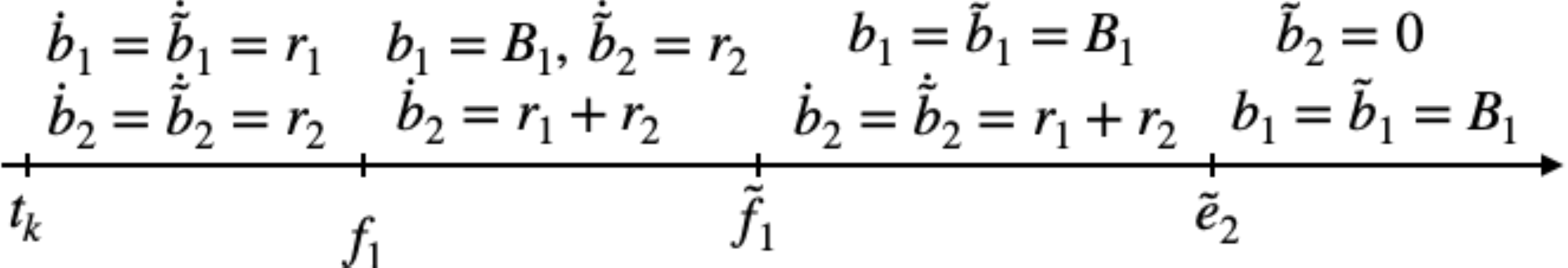}
          \caption{$\boldsymbol{r_2<0<r_1,\ r_1 < |r_2|}$; when full, neither $\boldsymbol{b_1}$ nor $\boldsymbol{\tilde{b}_1}$ allows discharge}
        \end{figure}
        \noindent {\bf 1b-C:} When full, $b_1$ does not allow the
        discharge but $\tilde{b}_1$ does i.e.
        $r_1 \geq \min(c_1,-r_2)$ and $r_1 < \min(c_1+\epsilon,-r_2).$
        Possible event sequences are:
        \begin{enumerate}
        \item $f_1 < \tilde{e}_2 < \tilde{e}_1 < e_2$
        \item $f_1 < \tilde{e}_2 < e_2 < \tilde{e}_1$
        \end{enumerate}
        In these cases, note that from $f_1$ onwards, $b_1$ shares
        with $b_2$ with rate $r_1(>0)$ due to overflow. Since
        $r_1<|r_2|$, it is easy to see that either both batteries,
        $b_2$ and $\tilde{b}_2$, discharges ($b_2$ discharges slowly
        due to the overflow it receives) or eventually at one point
        $\tilde{b}_2$ hits the zeros level and remains zero after
        that. This ensures that
        $b_2(t) \geq \tilde{b}_2(t)\ \forall t\in [f_1, t_{k+1}]$. The
        fact that $b_1$ remains full from $f_1$ onwards, ensures that
        $b_1(t) \geq \tilde{b}_1(t)\ \forall t\in [f_1,
        t_{k+1}]$. This is depicted below.
        \begin{figure}[H]
          \includegraphics[trim={0 0 0 0},scale=0.45]{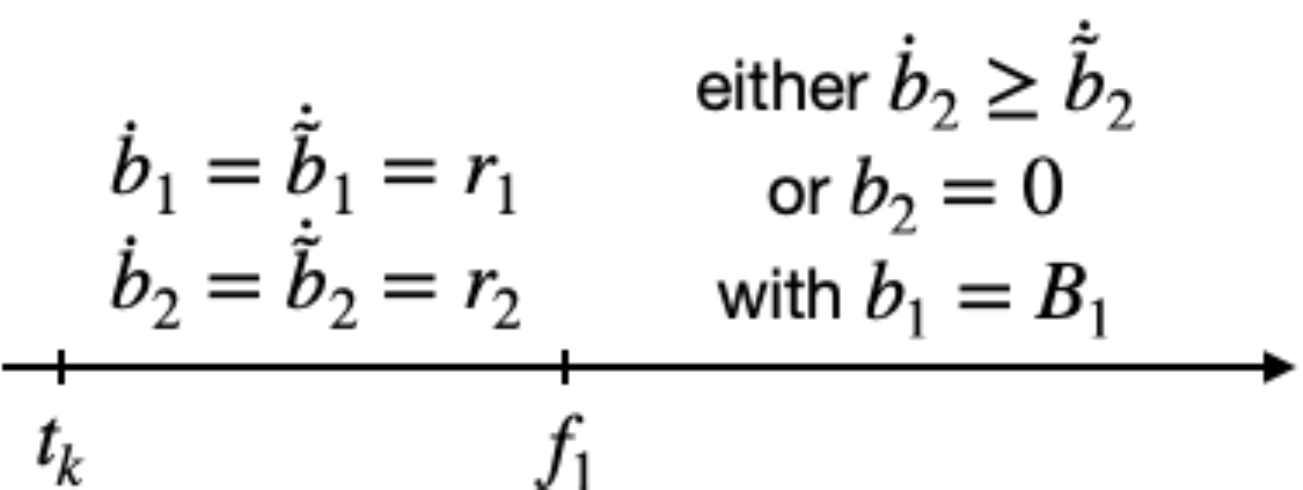}
          \caption{$\boldsymbol{r_2<0<r_1,\ r_1 < |r_2|}$; when full, $\boldsymbol{b_1}$ doesn't allow discharge but $\boldsymbol{\tilde{b}_1}$ does}
        \end{figure}
      

      \noindent {\bf Subcase 2:} $\tilde{b}_2$ drains to zero before
      $b_1$ becomes full ($t_k < \tilde{e}_2 < f_1$)

      \noindent {\bf 2a:} $r_1\geq |r_2|$\\
      From $\tilde{e}_2$ onwards, $\tilde{b}_1$ shares energy to
      $\tilde{b}_2$ with rate constrained to
      $\min\{c_1+\epsilon, (r_2)_{-}\}$. The fact that
      $r_1 \geq |r_2|$ implies that $\tilde{b}_1$ will continues to
      charge even after sharing.
        \begin{figure}[H]
          \includegraphics[trim={0 0 0 0},scale=0.45]{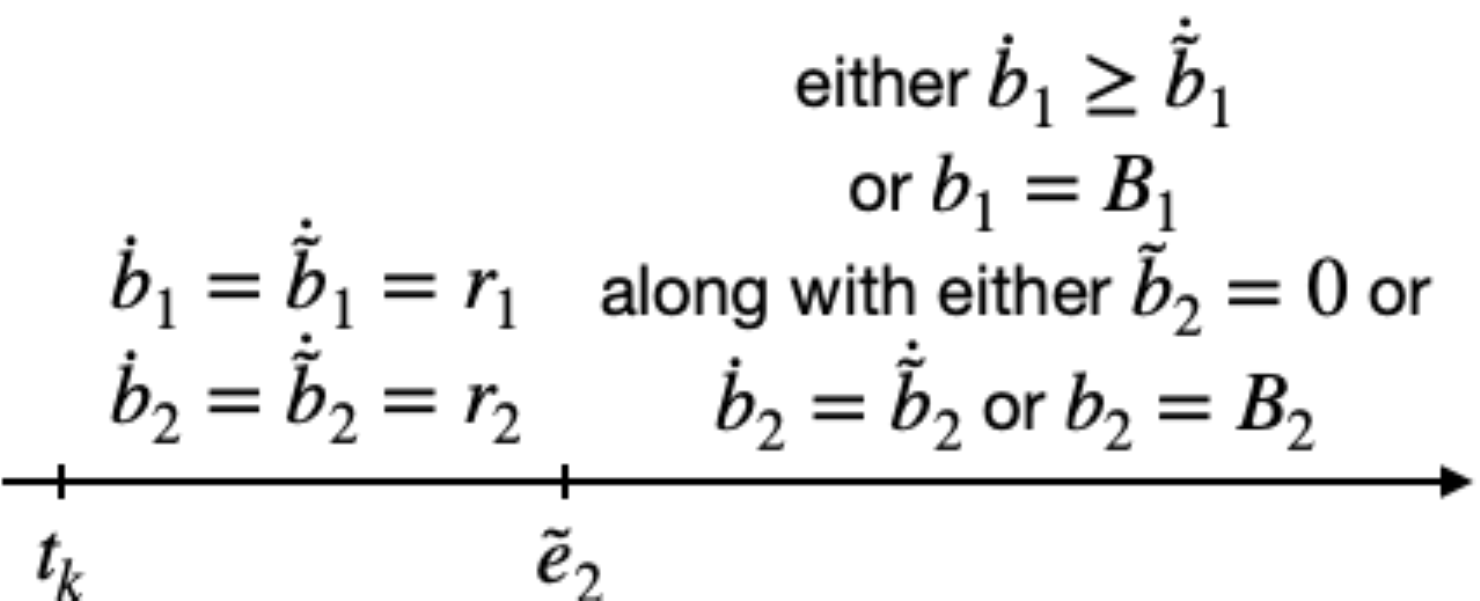}
          \caption{$\boldsymbol{r_2<0<r_1,\ r_1 \geq |r_2|}$}
        \end{figure}
        \noindent {\bf 2b:} $r_1<|r_2|$\\
        Since $r_1<|r_2|,$ $\tilde{b}_2$ remains empty from
        $\tilde{e}_2$ onwards, and hence
        $\tilde{b}(t)\leq b_2(t)\ \forall t\in[\tilde{e}_2,\ t_{k+1}]$
        which proves Lemma~\ref{lem:bat_induction} for Agent~2 side
        which comprises $b_2$ and $\tilde{b}_2$ (see
        Figure~\ref{fig:Case 4(2b)}).
        \begin{figure}[H]
          \includegraphics[trim={0 0 0 0},scale=0.45]{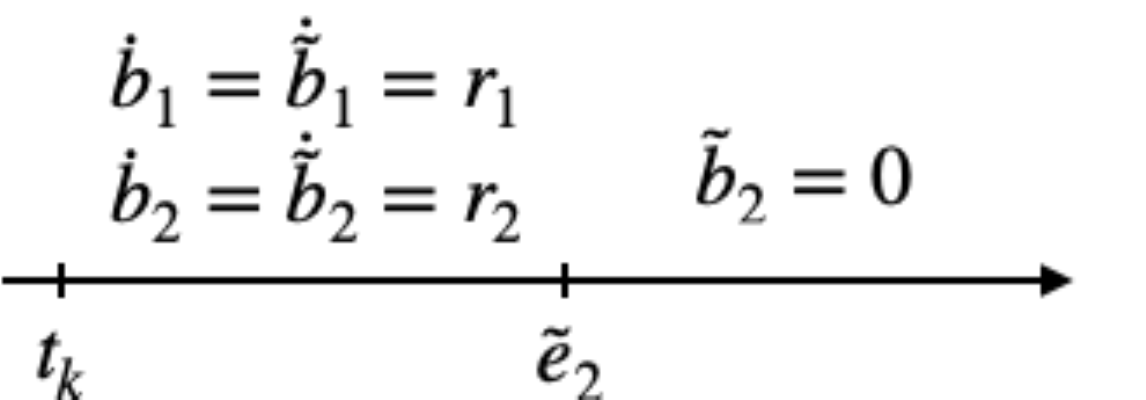}
          \caption{$\boldsymbol{r_2<0<r_1,\ r_1 < |r_2|}$}
          \label{fig:Case 4(2b)}
        \end{figure}
        
        On Agent~1 side, $\tilde{b}_1$ shares to $\tilde{b}_2$ with a
        maximum sharing rate $\min (c_1+\epsilon,(r_2)_{-})$.
        %
        Here, as before, there are three cases to consider.

        \noindent {\bf 2b-A:} When full, both $b_1$ and $\tilde{b}_1$
        allows discharge i.e. $r_1<\min(c_1,-r_2)$ and hence
        $r_1<\min(c_1+\epsilon,-r_2).$ Possible event sequences are:
          \begin{enumerate}
          \item $\tilde{e}_2<f_1<e_2<\tilde{e}_1<e_1$
          \item $\tilde{e}_2<f_1<\tilde{e}_1<e_2<e_1$
          \item $\tilde{e}_2<e_2<\tilde{e}_1<e_1$ 
          \item $\tilde{e}_2<\tilde{e}_1<e_2<e_1$
          \end{enumerate}
          \noindent {\bf 2b-B:} When full, neither $b_1$ nor
          $\tilde{b}_1$ allows the discharge i.e.
          $r_1 \geq \min(c_1,-r_2)$ and
          $r_1 \geq \min(c_1+\epsilon,-r_2).$ Possible event sequences
          are:
          \begin{enumerate}
          \item $\tilde{e}_2<f_1<e_2<\tilde{f}_1$
          \item $\tilde{e}_2<f_1<\tilde{f}_1<e_2$
          \item $\tilde{e}_2<e_2<f_1<\tilde{f}_1$ 
          \end{enumerate}
          
          \noindent {\bf 2b-C:} When full, $b_1$ don't allow the
          discharge but $\tilde{b}_1$ does i.e.
          $r_1 \geq \min(c_1,-r_2)$ and
          $r_1 < \min(c_1+\epsilon,-r_2).$ Possible event sequences
          are:
          \begin{enumerate}
          \item $\tilde{e}_2<f_1<\tilde{e}_1<e_2$
          \item $\tilde{e}_2<f_1<e_2<\tilde{e}_1$
          \item $\tilde{e}_2<\tilde{e}_1<f_1<e_2$
          \item $\tilde{e}_2<\tilde{e}_1<e_2<f_1$
          \item $\tilde{e}_2<e_2<f_1<\tilde{e}_1$
          \item $\tilde{e}_2<e_2<\tilde{e}_1<f_1$
          \end{enumerate}
        
          With careful observation along each of these sequences, it can be easily
          verified that
          $\tilde{b}_1(t) \leq b_1(t),\ \tilde{b}_2(t)\leq b_2(t)\
          \forall t \in [\tilde{e}_2\ t_{k+1}].$

 	}
\end{caseof}
\end{proof}

\end{document}